\documentclass[10pt,a4paper]{article}

\usepackage{amsmath,amsgen,latexsym}
\usepackage{amstext,amssymb,amsfonts,latexsym}
\usepackage{theorem}
\usepackage{pifont}
\usepackage[dvips]{graphics,epsfig}
\usepackage[dvips]{graphicx}

 \newcommand{\bs}{\bigskip}
 \newcommand{\ms}{\medskip}
 \newcommand{\n}{\noindent}
 \newcommand{\s}{\smallskip}
 \newcommand{\hs}[1]{\hspace*{ #1 mm}}
 \newcommand{\vs}[1]{\vspace*{ #1 mm}}

 \newcommand{\nat}{\mathbb{N}}

 \newcommand{\prob}{{\mathrm{Prob}}}

 \newcommand{\ie}{\textrm{i.e.},\hspace*{2mm}}
 \newcommand{\eg}{\textrm{e.g.},\hspace*{2mm}}
 
 \newcommand{\etalc}{\textrm{et al.}}

 \newcommand{\AAA}{{\cal A}}
  
 \newcommand{\BB}{{\cal B}}

 \newcommand{\HH}{{\cal H}}

 \newcommand{\UU}{{\cal U}}

 \newcommand{\np}{\mathrm{NP}}

 \def\bbox{\vrule height6pt width6pt depth1pt}

\theoremstyle{plain}
\theoremheaderfont{\bfseries}
 \newtheorem{theorem}{Theorem}[section]
 \newtheorem{lemma}[theorem]{Lemma}
 
 \newtheorem{corollary}[theorem]{Corollary}

 {\theorembodyfont{\rmfamily}
  \newtheorem{definition}[theorem]{Definition}}
 {\theorembodyfont{\rmfamily} }
 {\theorembodyfont{\rmfamily} }

 \newtheorem{claim}{Claim}

 \newenvironment{proof}{\par \noindent
            {\bf Proof. \hs{2}}}{\hfill$\Box$ \vspace*{3mm}}

 \newenvironment{proofof}[1]{\vspace*{5mm} \par \noindent
         {\bf Proof of #1.\hs{2}}}{\hfill$\Box$ \vspace*{3mm}}

 \newcommand{\ceilings}[1]{\lceil #1 \rceil}

 \newcommand{\qubit}[1]{| #1 \rangle}
 \newcommand{\bra}[1]{\langle #1 |}
 \newcommand{\ket}[1]{| #1 \rangle}
 \newcommand{\measure}[2]{\langle #1 | #2 \rangle}
 \newcommand{\density}[2]{| #1 \rangle\!\langle #2 |}
 \newcommand{\trace}{\mathrm{tr}}

\newcommand{\ignore}[1]{}

\begin{document}
\pagestyle{plain}
\pagenumbering{arabic}
\setcounter{page}{1}
\begin{center}
{\Large A Non-Interactive Quantum Bit Commitment Scheme that Exploits the
Computational Hardness of Quantum State Distinction} \bs\\
{\sc Tomoyuki Yamakami}\footnote{Present Affiliation: Department of Information Science, University of Fukui, 3-9-1 Bunkyo, Fukui 910-8507, Japan} \ms\\
\end{center}

\begin{quote}
{\small
\n{\bf Abstract:}
We propose an efficient quantum protocol performing quantum bit commitment, which is a simple cryptographic primitive involved with two parties, called a committer and a verifier.
Our protocol is non-interactive, uses no supplemental shared information, and achieves computational concealing and statistical binding under a natural complexity-theoretical assumption. An earlier protocol in the literature relies on the existence of an efficient quantum one-way function. Our protocol, on the contrary, exploits a seemingly weaker assumption on  computational difficulty of distinguishing two specific ensembles of reduced quantum states. This assumption is guaranteed by, for example,
computational hardness of solving the graph automorphism problem efficiently on a quantum computer.

\s
\n{\bf Keywords.} quantum bit commitment, quantum computation, distinction problem, graph automorphism, computational concealing, statistical binding
}
\end{quote}

\section{Introduction}

{\em Bit commitment} is a fundamental cryptographic primitive between two parties and its schemes have been applied to build other useful cryptographic protocols, including secure coin flipping, zero-knowledge proofs, secure multiparty computation, signature schemes, and secret sharing. A protocol for the bit commitment demands the following two security notions: {\em concealing} and {\em binding}. In a committing phase, Alice (committer) first commits a bit and sends Bob (verifier) its encrypted information, from which Bob cannot decipher her bit. In an opening phase, she reveals her bit; however, Bob can detect her wrongdoing if she presents the bit different from what she had committed in the earlier phase.

A quantum key distribution scheme \cite{BB84} is well-known to be unconditionally secure, whereas it is proven that no quantum bit commitment scheme  achieves unconditional security \cite{LC97,May97}. Chailloux and Kerenidis \cite{CK11} recently argued that no protocol for quantum bit commitment achieves a cheating probability of less than $0.739$.
These facts immediately prompt us to seek a reasonable means to build practically durable protocols for quantum bit commitment.
Technological limitations of current quantum device, on one hand, have been  used to design feasible protocols in, \eg \cite{DFSS05,DV12}.
Dumais, Mayers, and Salvail \cite{DMS00}, on the other hand, used a computationally difficult problem to construct a (non-interactive) protocol for quantum bit commitment. Their protocol requires the total communication cost of $O(n)$ qubits, where $n$ is a security parameter, and the security of the protocol relies on the existence of {\em quantum one-way permutation} (namely, a function that permutes a given set of strings with the one-way property that the function is easily computed buy is hard to be inverted).
In particular, the binding condition for the protocol is proven as follows.
If the condition does not hold, then Alice must have a strategy to deceive Bob. Using her strategy, we can efficiently invert a given quantum one-way permutation on a quantum computer, leading to a contradiction.
This protocol was later extended by Tanaka \cite{Tan03} to quantum string commitment using an additional technique of quantum fingerprinting to reduce  the communication cost between Alice and Bob. Recently, Koshiba and Odaira \cite{KO11} reduced this assumption to the existence of quantum one-way functions.

Whether a quantum one-way permutation exists is still an open question and it seems quite difficult to settle down the question. Naturally, we can ask if a use of quantum one-way permutation can be replaced by any other (seemingly weaker) computational assumption. In this paper, we look for other means to construct a quantum bit commitment protocol; in particular, we are interested in a computational problem of distinguishing between two ensembles of quantum states. This type of problem has been used to guarantee the security of quantum protocols. Chailloux, Kerenidis, and Rosgen \cite{CKR11} drew from
a slightly technical assumption of $\mathrm{QSZK}\nsubseteq \mathrm{QMA}$,\footnote{This statement means that there exists a quantum statistical zero-knowledge proof system that cannot be expressed as a form of quantum Merlin-Arthur proof system.}
which is not known to be true, a conclusion that a scheme for  ``auxiliary-input'' quantum bit commitment (which allows Alice and Bob to apply the same POVM operations during the committing phase) exists.
The purpose of this paper is to present a new scheme for quantum bit commitment with {\em no} such auxiliary inputs.

In 2005, Kawachi, Koshiba, Nishimura, and Yamakami \cite{KKNY12} devised two special ensembles of (reduced) quantum states and, from these ensembles,
they built a {\em quantum public-key cryptosystem} whose security relies on a computational assumption that the ensembles are hard to distinguish efficiently. These ensembles posses quite useful properties (stated in Section \ref{sec:permutation}), which are interesting on its own light and have been sought for other applications. As such an application, we actually use the two ensembles to build the aforementioned new scheme for quantum bit commitment.
Our scheme, given in Section \ref{sec:new-protocol}, is non-interactive, uses no auxiliary inputs, and achieves {\em computational concealing} and {\em statistical binding} at communication cost of $O(n\log{n})$ qubits.
Those security conditions of our scheme follows from an assumption on  computational hardness of distinguishing the two ensembles. Note that, if this hardness assumption fails to hold, then, for example, we can efficiently solve on a quantum computer  a classical problem, known as the {\em graph automorphism problem} (GA), in which we are asked to determine whether a given undirected graph is isomorphic to itself. This problem is not yet known to be either polynomial-time solvable or $\np$-complete. More importantly, our scheme has a concrete, explicit description, independent of the correctness of the assumption, and potentially it might be applied to other fields.

The computationally concealing condition for our scheme follows directly from the indistinguishability of two encrypted quantum states produced for two different committed bits $0$ and $1$. The statistical binding condition is met by an application of {\em state partitioning}, which is a means to partition a given quantum state into two specific orthogonal states. The details of the security conditions will be given in Section \ref{sec:analysis-protocol}.

\section{Main Theorem}\label{sec:main-contribution}

Throughout this paper, we will work with various finite dimensional Hilbert spaces. For instance, $\HH_{2}$ denotes the $2$-dimensional Hilbert space  spanned by a binary basis $\{\ket{0},\ket{1}\}$; that is, $\HH_{2} =\mathrm{span}\{\qubit{0},\qubit{1}\}$. We use Dirac's ket notation $\ket{\phi}$ to express quantum states and $\bra{\phi}$ for the conjugate transpose of $\ket{\phi}$. The notation $\measure{\phi}{\psi}$ expresses the inner product of $\ket{\phi}$ and $\ket{\psi}$. The {\em norm} of $\ket{\phi}$ is given as $\|\ket{\phi}\| =\sqrt{\measure{\phi}{\psi}}$. An {\em orthogonal measurement} (or von Neumann measurement) of a quantum state is described by a set of orthogonal projections acting on a given Hilbert space.

We will use informal term of {\em quantum algorithms} to describe transformations of quantum bits (or qubits) throughout this paper. A quantum algorithm has been often modeled by a mechanical device of {\em quantum Turing machine} \cite{BV97} or {\em quantum circuit families} \cite{Yao93}. We are particularly interested in quantum algorithms that terminate within a polynomial number of steps with respect to the size of input instances. We call such algorithms {\em polynomial-time algorithms}.

We choose the following definition for the computationally concealing condition, because this captures a more intuitive notion of Bob being unable  to gaining any significant amount of information out of Alice. As we will show in Lemma \ref{distinguish-bound}, this definition comes from the indistinguishability between two quantum states sent from Alice.

In a committing phase, Alice encodes her committed bit $a$ into a quantum state and sends its (possibly reduced) state $\chi_a$ to Bob. We demand the following security against Bob.

\begin{definition}(computational concealing)
A non-interactive quantum bit commitment scheme is {\em computationally concealing} if, for any positive polynomial $p$, there is no polynomial-time quantum algorithm that outputs $a$ from instance $\chi_a$  with error probability at least $1/2+1/p(n)$ for any $n\in\nat^{+}$.
\end{definition}

The security against Alice requires the following notion of statistical binding. In the committing phase, Alice starts with $\ket{0}$. She applies a quantum transformation $U_1$ and sends a subsystem $\HH_{commit}$. At the beginning of an {\em opening phase} (or a {\em revealing phase}), Alice applies $U_2^{(a)}$, where $a\in \{0,1\}$, if she wants to convince Bob that her committed bit is $a$. Alice's cheating strategy is modeled by a triplet $\UU=(U_1,U_2^{(0)},U_2^{(1)})$. Let $T_a^{(\UU)}(n)$ be the probability that Bob convinces himself that $a$ is truly a committed bit, provided that Bob faithfully follows the scheme.

\begin{definition}(statistical binding)
A non-interactive quantum bit commitment scheme is {\em statistically binding} if there exists a negligible function $\varepsilon(n)$ such that, for any cheating strategy  $\UU=(U_1,U_2^{(0)},U_2^{(1)})$ of Alice, $T_0^{(\UU)}(n)+T_1^{(\UU)}(n)\leq 1+  \varepsilon(n)$ holds for every length $n\in\nat^{+}$.
\end{definition}

Our main theorem concerns the notion of indistinguishable ensembles of quantum states which are generated efficiently. First, we introduce the necessary terminology to explain the statement.

It is time to introduce extra notions and notations. Let $\nat^{+}$ denote the set of all positive integers and set $\nat=\{0\}\cup\nat^{+}$. An ensemble $\{\rho(n)\}_{n\in\nat^{+}}$ of (reduced) quantum states is said to be {\em efficiently generated} if there exist two polynomially-bounded\footnote{A function $f:\nat^{+}\rightarrow\nat$ is polynomially bounded if there exists a (positive) polynomial $p$ for which $f(n)\leq p(n)$ for all $n\in\nat^{+}$.} functions $q,\ell:\nat^{+}\rightarrow\nat^{+}$ and a polynomial-time quantum algorithm $\AAA$ such that, on every input $\ket{1^n}\ket{0}$ ($n\in\nat^{+}$), (1) $\AAA$ generates $\ket{1^n}\ket{\Phi}$ of $q(n)$ qubits  and (2) $\rho(n)$ is obtained by tracing out the first $\ell(n)$ qubits of $\ket{\Phi}$; in notation, $\rho(n) =\trace_{\ell(n)}(\density{\Phi}{\Phi})$.
Let $\{\rho(n)\}_{n\in\nat^{+}}$ and $\{\chi(n)\}_{n\in\nat^{+}}$ be two ensembles of (reduced) quantum states. We say that a quantum algorithm $\AAA$ {\em distinguishes between $\{\rho(n)\}_{n\in\nat^{+}}$ and $\{\chi(n)\}_{n\in\nat^{+}}$ with advantage $\delta(n)$} \cite{KKNY12} if, for every $n\in\nat^{+}$, it holds that
\begin{equation}\label{eqn:def-advantage}
\delta(n) = \left| \prob[\AAA(1^n,\rho(n))=1] -  \prob[\AAA(1^n,\chi(n))=1] \right|.
\end{equation}
The succinct notation $\AAA(1^n,\rho(n))$ in Eq.~(\ref{eqn:def-advantage}) formally expresses $\AAA(\density{1^n}{1^n}\otimes \rho(n))$.

A function $\mu:\nat\rightarrow[0,1]$ is called {\em negligible} if, for any positive polynomial $p$, $\mu(n)\leq 1/p(n)$ holds for all but finitely many numbers $n\in\nat$.
If if there exists a polynomial-time quantum algorithm that distinguishes between two  ensembles $\{\rho(n)\}_{n\in\nat^{+}}$ and $\{\chi(n)\}_{n\in\nat^{+}}$ of (reduced) quantum states with non-negligible advantage, then the two ensembles are said to be {\em efficiently indistinguishable}. Otherwise, they are called {\em efficiently indistinguishable}.

\begin{theorem}\label{main_theorem}
{\rm (main theorem)}
There exists a pair of efficiently generated ensembles of (reduced) quantum states such that (1) they are efficiently indistinguishable and (2) from them, we can construct a non-interactive quantum bit commitment with  computationally canceling and statistically binding conditions.
\end{theorem}

In Section \ref{sec:permutation}, we will explicitly define an ensemble pair described in Theorem \ref{main_theorem}. With help of \cite[Theorem 2.5]{KKNY12}, the existence of such an ensemble pair is guaranteed if the {\em graph automorphism problem} (GA) is difficult to solve efficiently on any  quantum computer. As an immediate consequence of the theorem, we obtain the following statement.

\begin{corollary}\label{main_corollary}
If no polynomial-time quantum algorithm solves $\mathrm{GA}$ with error probability at least $2^{-n}$, where $n$ is the vertex set size if an input graph, then there exists a non-interactive quantum bit commitment with the conditions of computationally canceling and statistically binding.
\end{corollary}

In the subsequent sections, we will prove Theorem \ref{main_theorem}.

\section{Quantum Bit Commitment}\label{sec:bit-commitment}

We will present a scheme for non-interactive quantum bit commitment. Our scheme is based on an ensemble of special quantum states, introduced in \cite{KKNY12}. In Section \ref{sec:permutation}, we will explain these quantum states. A useful property of state partitioning will be discussed in Section \ref{sec:partition}. Finally, we will present in Section \ref{sec:new-protocol} our protocol of quantum bit commitment between Alice and Bob.

\subsection{Special Quantum States with Hidden Permutations}
\label{sec:permutation}

As a preparation to our quantum bit commitment scheme, we will introduce an ensemble of special quantum states given in \cite{KKNY12}.

Let $n$ be any number in $\nat^{+}$, which is used as a {\em security parameter}; for our purpose, we demand that $n$ is even and $n/2$ is odd.  Let $S_n$ denote the set of all permutations $\sigma:[n]\rightarrow[n]$, where $[n]$ is the integer set $\{1,2,3,\ldots,n\}$. Since $|S_n|=n!$, every element in $S_n$ can be expressed using at most $\ceilings{\log(n!)}$ ($= O(n\log{n})$) qubits. The special set $K_n$ is a subset of $S_n$, consisting only of $\pi$ satisfying $\pi\pi=\mathrm{id}$ and $\pi(i)\neq i$ for all $i\in[n]$.

Given three elements $s\in\{0,1\}$, $\sigma\in S_n$, and $\pi\in K_n$, we define a useful quantum state $\ket{\phi_{\sigma,s}^{(\pi)}}$ as
\[
\ket{\phi_{\sigma,s}^{(\pi)}(n)} = \frac{1}{\sqrt{2}} (\ket{\sigma}+(-1)^s\ket{\sigma\pi}).
\]
For each permutation $\pi$ in $K_n$, we partition $S_n$ into two subsets $\hat{S}_0$ and $\hat{S}_1$, which satisfy the following condition:
for every index $a\in\{0,1\}$ and for all elements $\sigma\in S_n$, $\sigma\in \hat{S}_a$ implies $\sigma\pi\in \hat{S}_{1-a}$.
Let $S_n^{(\pi)}$ denote  one of these subsets of $S_n$ that contains $\mathrm{id}$.
Notice that $S_n^{(\pi)}$ is uniquely determined from $\pi$.
It is easily seen that the set $\BB^{(\pi)} = \{\ket{\phi_{\sigma,s}^{(\pi)}} \mid \sigma \in S_n^{(\pi)},s\in\{0,1\}\}$ forms a computational basis for the Hilbert space $\HH_{S_n}=\mathrm{span}\{\ket{\sigma}\mid \sigma\in S_n\}$.

In what follows, we fix $n\in\nat^{+}$ and $\pi\in K_n$. For each bit  $s\in\{0,1\}$, we define the quantum state
\[
\rho_{s}^{(\pi)}(n) = \frac{1}{|S_n|} \sum_{\sigma\in S_n} \density{\phi_{\sigma,s}^{(\pi)}(n)}{\phi_{\sigma,s}^{(\pi)}(n)}.
\]
Notice that the quantum states $\rho_{0}^{(\pi)}(n)$ and $\rho_{1}^{(\pi)}(n)$ are respectively denoted $\rho_{\pi}^{+}(n)$ and $\rho_{\pi}^{-}(n)$ in \cite{KKNY12}.
For convenience, we use the notation $\ket{\Phi_{s}^{(\pi)}(n)}$ to denote a pure quantum state $\frac{1}{\sqrt{|S_n|}}\sum_{\sigma\in S_n}\ket{\sigma}\ket{\phi_{\sigma,s}^{(\pi)}(n)}$,
which is a {\em purification} of $\rho_{s}^{(\pi)}$, because $\rho_{s}^{(\pi)}(n)$ coincides with  the partial trace $\trace_{1}(\density{\Phi_{s}^{(\pi)}(n)}{\Phi_{s}^{(\pi)}(n)})$, where $\trace_{1}$ is the {\em partial trace} over the first register (that is, the operator tracing out the first register). It is also useful to note that $\sum_{\sigma\in S_n}\ket{\sigma}\ket{\phi_{\sigma,s}^{(\pi)}} = \sum_{\sigma\in S_n}\ket{\phi_{\sigma,s}^{(\pi)}}\ket{\sigma}$; thus, it holds that
 $\trace_{1}(\density{\Phi_{s}^{(\pi)}(n)}{\Phi_{s}^{(\pi)}(n)}) = \trace_{2}(\density{\Phi_{s}^{(\pi)}(n)}{\Phi_{s}^{(\pi)}(n)})$.

To make our notation simple, we hereafter omit ``$n$'' whenever ``$n$'' is clear from the context.

Note that any quantum state $\ket{\gamma}$ in $\HH_{S_n}$ can be expressed as $\sum_{a\in\{0,1\}}\sum_{\sigma\in S_n}\alpha_{a,\sigma,\pi} \ket{\phi_{\sigma,a}^{(\pi)}}$ for any fixed permutation $\pi\in K_n$.
Basic properties of $\ket{\phi_{\sigma,s}^{(\pi)}}$ and $\ket{\Phi_{s}^{(\pi)}}$ are summarized in the following lemma. In the lemma, we conveniently use two notations ``$\wedge$'' and ``$\vee$'' to mean the logical connectives ``AND'' and ``OR'',  respectively.

\begin{lemma}\label{base-phi-1}
Let $n\in\nat^{+}$, $s\in\{0,1\}$, $\pi\in K_n$, and $\sigma,\tau\in S_n$.
\renewcommand{\labelitemi}{$\circ$}
\begin{enumerate}\vs{-1}
  \setlength{\topsep}{-2mm}%
  \setlength{\itemsep}{1mm}%
  \setlength{\parskip}{0cm}%

\item $\ket{\phi_{\sigma\pi,s}^{(\pi)}} = (-1)^s \ket{\phi_{\sigma,s}^{(\pi)}}$.

\item $\measure{\phi_{\sigma,0}^{(\pi)}}{\phi_{\sigma,1}^{(\pi)}} =0$.

\item $\measure{\phi_{\sigma,s}^{(\pi)}}{\phi_{\tau,s}^{(\pi)}}= 1$ if $\tau=\sigma$; $(-1)^s$ if $\tau=\sigma\pi$; $0$ otherwise.

\item $\measure{\phi_{\sigma,0}^{(\pi)}}{\phi_{\tau,0}^{(\kappa)}} = 1$ if
    $\pi=\kappa\wedge (\sigma=\tau \vee \sigma=\tau\pi)$; $\frac{1}{2}$ if $\pi\neq \kappa\wedge (\sigma=\tau\vee \sigma=\tau\kappa \vee \sigma=\tau\pi)$; $0$ otherwise.

\item $\measure{\phi_{\sigma,1}^{(\pi)}}{\phi_{\tau,0}^{(\kappa)}} = \frac{1}{2}$ if $\pi\neq \kappa \wedge (\sigma=\tau \vee \sigma=\tau\kappa)$; $-\frac{1}{2}$ if $\pi\neq \kappa \wedge (\sigma\pi=\tau \vee \sigma\pi=\tau\kappa)$; $0$ otherwise.
\end{enumerate}\vs{-2}
\end{lemma}

\begin{proof}
(1) Since $\sigma\pi\pi=\sigma$, it follows that  $\sqrt{2}\ket{\phi_{\sigma,s}^{(\pi)}} = \ket{\sigma\pi\pi}+(-1)^{s}\ket{\sigma\pi} = (-1)^s[\ket{\sigma\pi} +(-1)^s\ket{\sigma\pi\pi} ] = (-1)^s\sqrt{2}\ket{\phi_{\sigma\pi,s}^{(\pi)}}$. This implies that $\ket{\phi_{\sigma,s}}=(-1)^{s}\ket{\phi_{\sigma\pi,s}}$, leading to the desired consequence.

(2) For simplicity, we write $P$ for $\measure{\phi_{\sigma,0}^{(\pi)}}{\phi_{\sigma,1}^{(\pi)}}$. Note that $\pi\in K_n$ implies $\sigma\pi\neq\sigma$ because $\sigma\pi=\sigma$ is equivalent to $\pi=\mathrm{id}$. Since $2P = (\bra{\sigma}+\bra{\sigma\pi})(\ket{\sigma}-\ket{\sigma\pi}) = \measure{\sigma}{\sigma} + \measure{\sigma\pi}{\sigma} - \measure{\sigma}{\sigma\pi} - \measure{\sigma\pi}{\sigma\pi}$, $2P$ equals $0$.

(3) Consider the case of $s=0$. Let $P= \measure{\phi_{\sigma,s}^{(\pi)}}{\phi_{\tau,s}^{(\pi)}}$. Note that $2P = 2(\measure{\sigma}{\tau}+\measure{\sigma}{\tau\pi})$ since $\measure{\sigma\pi}{\tau\pi} = \measure{\sigma}{\tau}$. Since $\pi\neq\mathrm{id}$, $\measure{\sigma}{\tau}=1$ implies $\measure{\sigma}{\tau\pi}=0$ and also  $\measure{\sigma}{\tau\pi}=1$ implies $\measure{\sigma}{\tau}=0$. Thus, if either $\tau=\sigma$ or $\tau=\sigma\pi$, we have $2P=2$; otherwise, $2P=0$. The other case of $s=1$ is similarly handled.

(4) By setting $P= \measure{\phi_{\sigma,0}^{(\pi)}}{\phi_{\tau,0}^{(\kappa)}}$, we obtain
$2P = \measure{\sigma}{\tau}+\measure{\sigma}{\tau\kappa} + \measure{\sigma\pi}{\tau} + \measure{\sigma\pi}{\tau\kappa}$. If $\pi=\kappa$, then (3) implies the desired result. Now, assume that $\pi\neq \kappa$. If $\sigma=\tau$, then $\measure{\sigma}{\tau}=1$ and the other terms in the above expansion of $2P$ are all zeros, because $\pi,\kappa\neq \mathrm{id}$ and $\pi\neq\kappa$. From these results follows $2P=1$. Similarly, if $\sigma=\tau\kappa$, then $\measure{\sigma}{\tau\kappa}=1$ and the other terms are zeros; thus, we obtain $2P=1$. The remaining case of $\sigma=\tau\pi$ is similarly handled. When all the above-mentioned  cases fail, no terms in the expansion of $2P$ are $1$. Therefore, we conclude that $2P=0$.

(5) Let $P= \measure{\phi_{\sigma,1}^{(\pi)}}{\phi_{\tau,0}^{(\kappa)}}$. Note that $2P = \measure{\sigma}{\tau} - \measure{\sigma\pi}{\tau} + \measure{\sigma}{\tau\kappa} - \measure{\sigma\pi}{\tau\kappa}$. If $\pi=\kappa$, then (2) implies $2P=0$. In what follows, we assume that $\pi\neq\kappa$. If $\sigma=\tau$, then it follows that  $2P=\measure{\sigma}{\tau}=1$
because $\pi,\kappa\neq\mathrm{id}$ and $\pi\neq\kappa$. Thus, we obtain  $2P=1$. Similarly, when $\sigma=\tau\kappa$, we obtain $2P= \measure{\sigma}{\tau\kappa}= 1$. From $\sigma\pi=\tau$, it follows that  $2P=-\measure{\sigma\pi}{\tau}=-1$. Finally, we note that $\sigma\pi=\tau\kappa$ implies $2P= -\measure{\sigma\pi}{\tau\kappa} = -1$.
\end{proof}

We give another useful property of $\ket{\phi_{\sigma,s}^{(\pi)}}$.  This property will play an important role in Section \ref{sec:distilation}.

\begin{lemma}\label{Phi-rewritten}
For fixed $\pi\in K_n$ and $\sigma\in S_n$, it holds that $\ket{\phi_{\sigma,1}^{(\pi)}} = \frac{1}{|K_n|-1}\sum_{\kappa\in K_n} (\ket{\phi_{\sigma,0}^{(\kappa)}} - \ket{\phi_{\sigma\pi,0}^{(\kappa)}})$.
\end{lemma}

\begin{proof}
Fix $\pi\in K_n$ and $\sigma\in S_n$. Note that the value $\sqrt{2}\sum_{\kappa\in K_n}(\ket{\phi_{\sigma,0}^{(\kappa)}} - \ket{\phi_{\sigma\pi,0}^{(\kappa)}})$ is $\sum_{\kappa\in K_n}((\ket{\sigma} - \ket{\sigma\pi}) + (\ket{\sigma\kappa} - \ket{\sigma\pi\kappa}))$, which equals $|K_n|(\ket{\sigma} - \ket{\sigma\pi}) + \sum_{\kappa\in K_n} (\ket{\sigma\kappa} - \ket{\sigma\pi\kappa})$.
The last term equals $\sum_{\kappa\in K_n}\ket{\sigma\kappa} - \sum_{\kappa\in K_n}\ket{\sigma\pi\kappa}$. Let us consider the function $f:K_n\rightarrow K_n\cup\{\mathrm{id}\}$ defined as $f(\kappa) = \pi\kappa$. This $f$ satisfies the following three properties: (i) $f$ is one-to-one, (ii) $f(\pi) = \mathrm{id}$, and (iii) there is no element $\kappa\in K_n$ satisfying $f(\kappa)=\pi$. We conclude that $f$ is a bijection on the restricted domain $K_n-\{\pi\}$. This fact implies that $\sum_{\kappa\in K_n} \ket{\sigma\kappa} - \sum_{\kappa\in K_n}\ket{\sigma\pi\kappa} = \ket{\sigma\pi} - \ket{\sigma}$.

Overall, $\sqrt{2}\sum_{\kappa\in K_n}(\ket{\phi_{\sigma,0}^{(\kappa)}} -  \ket{\phi_{\sigma\pi,0}^{(\kappa)}}) = |K_n|(\ket{\sigma} - \ket{\sigma\pi}) + (\ket{\sigma\pi} - \ket{\sigma})$, which equals $(|K_n|-1)(\ket{\sigma}-\ket{\sigma\pi})$; in other words,  $\sqrt{2}(|K_n|-1)\ket{\phi_{\sigma,1}^{(\pi)}}$. From this equality, the lemma follows immediately.
\end{proof}

Hereafter, we will give two quantum procedures, which are useful in the description of our quantum bit commitment scheme in Section \ref{sec:new-protocol}. First, we introduce several useful unitary operations.
The {\em Hadamard transform} $H$ acts on the system $\HH_{2}$ as $H\ket{s} = \frac{1}{\sqrt{2}}(\ket{0}+(-1)^s\ket{1})$ for every bit $s\in\{0,1\}$. The {\em controlled-$\pi$ operator} $C_{\pi}$ acts on $\HH_{2}\otimes \HH_{S_n}$ as $C_{\pi}\ket{a}\ket{\sigma}=\ket{a}\ket{\sigma\pi}$ if $a=1$, and $\ket{a}\ket{\sigma}$ otherwise. The {\em controlled-NOT$_{id}$ operator} $CNOT_{\mathrm{id}}$ acts on $\HH_{2}\otimes \HH_{S_n}$ as $CNOT_{\mathrm{id}}\ket{a}\ket{\sigma} =   (NOT\ket{a})\ket{\sigma}$ if $\sigma=\mathrm{id}$, and $\ket{a}\ket{\sigma}$ otherwise. Moreover, let $U_{sgn}$ denote a unitary operator mapping $\ket{\sigma}$ to $(-1)^{sgn(\sigma)}\ket{\sigma}$, where $\sigma\in S_n$ and $sgn(\sigma)$ is $1$ ($0$, resp.) if $\sigma$ is an even (odd, resp.) permutation in $S_n$.  The {\em controlled-SAWP} operator $C_{swap}^{(i,j)}$ (with $1\leq i<j\leq n$) exchanges the contents of the $i$th and $j$th registers among $n$ registers; that is, $C_{swap}^{(i,j)}\ket{s_1}\cdots \ket{s_i}\cdots \ket{s_j}\cdots  \ket{s_n} =  \ket{s_1}\cdots \ket{s_j}\cdots \ket{s_i}\cdots  \ket{s_n}$.

We will present two useful quantum transforms.

\s
{\sc [Procedure 1]}
The following procedure $P_1$ can generate the quantum state  $\HH = \ket{0}\ket{\sigma}\ket{\pi}\ket{\phi_{\sigma,0}^{(\pi)}}$ in the system $\HH_2\otimes \HH_{S_n}^{(1)}\otimes \HH_{S_n}^{(2)}\otimes \HH_{S_n}^{(3)}$ from $\ket{0}\ket{\sigma}\ket{\pi}\ket{\mathrm{id}}$
if $\pi\neq \mathrm{id}$.
To $\ket{\pi}\ket{0}\ket{\mathrm{id}}$ in the system $\HH_{S_n}^{(2)}\otimes \HH_{2}\otimes \HH_{S_n}^{(3)}$, we first apply two operators $I\otimes H\otimes I$ and $I\otimes C_{\pi}$, where $H$ is the Hadamard transform.
This process generates a quantum state  $\frac{1}{\sqrt{2}}\ket{\pi}(\ket{0}\ket{\mathrm{id}}+\ket{1}\ket{\pi})$.
Since $\pi\neq\mathrm{id}$, we apply $I\otimes CNOT_{\mathrm{id}}$ and generate the state
$\frac{1}{\sqrt{2}}\ket{\pi}\ket{0}(\ket{\mathrm{id}}+\ket{\pi})$. Now, consider  $\frac{1}{\sqrt{2}}\ket{\sigma}(\ket{\mathrm{id}}+\ket{\pi})$ in $\HH_{S_n}^{(2)}\otimes \HH_{S_n}^{(3)}$. Multiply $\sigma$ in the first register from the left to the second register, generating $\frac{1}{\sqrt{2}}\ket{\sigma}(\ket{\sigma}+\ket{\sigma\pi})$. In the end, we obtain $\ket{0}\ket{\sigma}\ket{\pi}\ket{\phi_{\sigma,0}^{(\pi)}}$ in $\HH$.

Similarly, we can generate $\ket{0}\ket{\pi}\ket{\Phi_0^{(\pi)}}$ from $\ket{0}\ket{\pi}\ket{0}\ket{\mathrm{id}}$ in $\HH_{2}\otimes \HH_{S_n}\otimes \HH_{S_n}\otimes \HH_{S_n}$ by running the following procedure $\tilde{P}_1$. After generating the state
$\frac{1}{\sqrt{2}}\ket{0}(\ket{\mathrm{id}}+\ket{\pi})$ as described above, we generate $\frac{1}{\sqrt{2|S_n|}}\sum_{\sigma\in S_n}\ket{\sigma}(\ket{\mathrm{id}}+\ket{\pi})$. Multiply each $\sigma$ in the first register to the content of the second register. We then  obtain
$\frac{1}{\sqrt{2|S_n|}}\ket{\sigma}(\ket{\sigma}+\ket{\sigma\pi})$, which is $\ket{\Phi_{0}^{(\pi)}}$.

\s

{\sc [Procedure 2]}
There is a simple procedure $P_2$ that transforms $\ket{\phi_{\sigma,s}^{(\pi)}}$ to $\ket{\phi_{\sigma,1-s}^{(\pi)}}$ without knowing $(s,\pi)$ as follows. Initially, we have $\ket{\phi_{\sigma,s}^{(\pi)}}$. We apply $U_{sgn}\otimes U_{sgn}$ to $\ket{\phi_{\sigma,s}^{(\pi)}}$. The resulted quantum state is $\frac{1}{\sqrt{2}}(-1)^{sgn(\sigma)}((-1)^{sgn(\sigma)}\ket{\sigma} + (-1)^{s+sgn(\sigma\pi)}\ket{\sigma\pi})$. Since $\pi$ is an odd permutation, this state equals $\frac{1}{\sqrt{2}}(-1)^{2sgn(\sigma)}(\ket{\sigma} +(-1)^{s+1}\ket{\sigma\pi})$, which is exactly $\ket{\phi_{\sigma,1-s}^{(\pi)}}$.
If we apply $I\otimes P_2$ to $\ket{\Phi_{s}^{(\pi)}}$, then we immediately obtain $\ket{\Phi_{1-s}^{(\pi)}}$.

\subsection{State Partitioning}\label{sec:partition}

Our quantum bit commitment protocol in Section \ref{sec:new-protocol} requires a method to ``partition'' a given quantum state $\chi$ in the system $\HH_{S_n}$ into two orthogonal states $\chi_0$ and $\chi_1$ that satisfy an extra property. A basic idea of state partitioning is inspired by a trapdoor property of \cite[Theorem 2.1]{KKNY12}.
Let $n\in\nat^{+}$ and $\chi$ be any reduced state in $\HH_{S_n}$.
This state can be expressed as $\chi = \chi_0+\chi_1$, where
$\chi_s = \sum_{\sigma\in S}p_{\sigma}\density{\phi_{\sigma,s}^{(\pi)}}{\phi_{\sigma,s}^{(\pi)}}$ for each bit $s\in\{0,1\}$.

\s
\n\hrulefill \s\\
\n{\sc State Partition Algorithm:} $C_{SPA}$
\renewcommand{\labelitemi}{$\circ$}
\begin{itemize}\vs{-1}
  \setlength{\topsep}{-2mm}%
  \setlength{\itemsep}{1mm}%
  \setlength{\parskip}{0cm}%

  \item[(S1)] Take an instance of the form $\chi' = \density{\pi}{\pi}\otimes \chi$ in a system $\HH_{S_n}\otimes \HH_{S_n}$. Prepare $\density{0}{0}\otimes \chi'$ in        $\HH = \HH_{2}\otimes \HH_{S_n} \otimes \HH_{S_n}$.

  \item[(S2)] Apply $H\otimes I^{\otimes 2}$. Since the second register of $\HH$  contains $\pi$, we can freely use the controlled-$\pi$ operator $C_{\pi}$. Here, we apply $C_{swap}^{(2,3)} (C_{\pi}\otimes I)C_{swap}^{(2,3)}$. Finally, apply $H\otimes I^{\otimes 2}$.

  \item[(S3)] The state $\density{0}{0}\otimes \chi'$ changes into $\density{0}{0}\otimes \density{\pi}{\pi} \otimes \chi_0 + \density{1}{1}\otimes \density{\pi}{\pi} \otimes \chi_1$. When we observe the first register, we find $0$ (resp., $1$) with probability exactly $\frac{1}{2}$.
\end{itemize}\vs{-2}
\n\hrulefill\\

Here, we will briefly discuss the correctness of the above algorithm.
Let $\chi$ be given at Step (S1). Assume that $\chi = \sum_{\sigma\in S_n}\sum_{s\in\{0,1\}} p_{\sigma,s}\density{\phi_{\sigma,s}^{(\pi)}}{\phi_{\sigma,s}^{(\pi)}}$. We introduce a new system $\HH\otimes \HH'$ and let a purification of $\chi$ in $\HH\otimes \HH' \otimes \HH_{S_n}$ be $\ket{\Phi} = \sum_{\sigma\in S_n}\sum_{s\in\{0,1\}} \sqrt{p_{\sigma,s}} \ket{s}\ket{\sigma} \ket{\phi_{\sigma,s}^{(\pi)}}$. Note that $\chi = \trace_{1,2}(\density{\Phi}{\Phi})$.
For each fixed $s$, we write $\ket{\psi_s} = \sum_{\sigma\in S_n}\sqrt{p_{\sigma,s}} \ket{s} \ket{\sigma} \ket{\phi_{\sigma,s}^{(\pi)}}$. Note that $\ket{\Phi} = \ket{\psi_0}+\ket{\psi_1}$, where $\chi_{s} = \trace_{1,2}(\density{\psi_s}{\psi_s}) = \sum_{\sigma\in S_n}\sqrt{p_{\sigma,s}}\density{\phi_{\sigma,s}^{(\pi)}}{\phi_{\sigma,s}^{(\pi)}}$ for each $s$.

Initially, we have a state $\ket{\Phi'} = \ket{0}\ket{\pi}\ket{\Phi}$, which equals $\sum_{\sigma\in S_n}\sum_{s\in\{0,1\}} \sqrt{p_{\sigma}} \ket{0} \ket{\pi}\ket{s} \ket{\sigma}\ket{\phi_{\sigma,s}^{(\pi)}}$.
Since we work on a purified state, it is convenient to expand $C_{\pi}$ to $\tilde{C}_{\pi}$ as follows. Let $\tilde{C}_{\pi} = C_{swap}^{(1,2)} C_{swap}^{(3,5)}(I \otimes C_{\pi} \otimes I^{\otimes 2})C_{swap}^{(3,5)}C_{swap}^{(1,2)}$. Step (S2) produces
$\ket{\Phi''} = (H\otimes I^{\otimes 4}) \tilde{C}_{\pi}(H\otimes I^{\otimes 4})\ket{\Phi'}$. By a direct calculation, the quantum state $(H\otimes I^{\otimes 4}) \tilde{C}_{\pi}(H\otimes I^{\otimes 4})\ket{0}\ket{\pi}\ket{\psi_s}$ equals
\[
\frac{1}{2}\sum_{\sigma\in S_n}\sqrt{p_{\sigma,s}} \ket{0} \ket{\pi} \ket{s} \ket{\sigma} \left( \ket{\phi_{\sigma,s}^{(\pi)}} + \ket{\phi_{\sigma\pi,s}^{(\pi)}} \right) +
\frac{1}{2}\sum_{\sigma\in S_n}\sqrt{p_{\sigma,s}} \ket{1} \ket{\pi} \ket{s} \ket{\sigma} \left( \ket{\phi_{\sigma,s}^{(\pi)}} - \ket{\phi_{\sigma\pi,s}^{(\pi)}} \right),
 \]
which coincides with $\ket{s}\ket{\pi}\ket{\psi_s}$ because  Lemma \ref{base-phi-1} implies $\ket{\phi_{\sigma,s}^{(\pi)}} = (-1)^s \ket{\phi_{\sigma\pi,s}^{(\pi)}}$.
Therefore, we conclude that $\ket{\Phi''} = \ket{0}\ket{\pi}\ket{\psi_0}+\ket{1}\ket{\pi}\ket{\psi_1}$.

\ms

Next, we want to examine the behavior of $C_{SPA}$ on instance $\density{\kappa}{\kappa}\otimes \chi$ given in Step (S1), where  $\kappa$ is different from $\pi$. In what follows, let $\kappa$ be any element in $K_n-\{\pi\}$. Recall that $\ket{\Phi} = \sum_{\sigma\in S_n} \sqrt{p_{\sigma}}\ket{s} \ket{\sigma}\ket{\phi_{\sigma,s}^{(\pi)}}$. Here, the quantum algorithm starts with $\ket{\Phi'_{\kappa}} = \ket{0}\ket{\kappa} \ket{\Phi}$. Following Step (S2), we calculate $(H\otimes I^{\otimes 4}) \tilde{C}_{\kappa}(H\otimes I^{\otimes 4})\ket{0}\ket{\pi}\ket{\psi_0}$. Note that, since the content of the second register is $\kappa$, we apply $C_{\kappa}$ instead of $C_{\pi}$.

The algorithm $C_{SPA}$ produces a quantum state
\[
\sum_{a\in \{0,1\}} \frac{1}{2\sqrt{2}}\sum_{\sigma\in S_n}\sqrt{p_{\sigma,s}} \ket{a} \ket{\pi} \ket{s} \ket{\sigma} \left( \left( \ket{\sigma} +(-1)^s\ket{\sigma\pi} \right) + (-1)^a \left( \ket{\sigma\kappa} + (-1)^s\ket{\sigma\pi\kappa} \right) \right).
 \]
This is equivalent to
\[
\sum_{a\in \{0,1\}} \frac{1}{2}\sum_{\sigma\in S_n}\sqrt{p_{\sigma,s}} \ket{a} \ket{\pi} \ket{s} \ket{\sigma} \left( \ket{\phi_{\sigma,a}^{(\kappa)}} + (-1)^s \ket{\phi_{\sigma\pi,a}^{(\kappa)}} \right).
\]
When $a=0$ and $s=0$, we obtain $\frac{1}{2}\sum_{\sigma\in S_n}\sqrt{p_{\sigma,s}} \ket{0} \ket{\pi} \ket{0} \ket{\sigma} (\ket{\phi_{\sigma,0}^{(\kappa)}} + \ket{\phi_{\sigma\pi,0}^{(\kappa)}})$.
Since $\measure{\phi_{\sigma,0}^{(\kappa)}}{\phi_{\sigma\pi,0}^{(\kappa)}} = 0$ by Lemma \ref{base-phi-1}(3), the norm of this state is
\[
\left\| \frac{1}{2}\sum_{\sigma\in S_n}\sqrt{p_{\sigma,s}} \ket{0} \ket{\pi} \ket{0} \ket{\sigma} \left( \ket{\phi_{\sigma,0}^{(\kappa)}} + \ket{\phi_{\sigma\pi,0}^{(\kappa)}} \right) \right\|^2 = \frac{1}{4}\sum_{\sigma\in S_n} p_{\sigma,0} \left( 2 + 2\measure{\phi_{\sigma,0}^{(\kappa)}}{\phi_{\sigma\pi,0}^{(\kappa)}} \right) = \frac{1}{2}\sum_{\sigma}p_{\sigma,0}.
\]
Similarly, when $a=0$ and $s=1$, the state  $\frac{1}{2}\sum_{\sigma\in S_n}\sqrt{p_{\sigma,1}} \ket{0} \ket{\pi} \ket{1} \ket{\sigma} (\ket{\phi_{\sigma,0}^{(\kappa)}} - \ket{\phi_{\sigma\pi,0}^{(\kappa)}})$ has norm $\frac{1}{2}\sum_{\sigma}p_{\sigma,1}$. By combining those values, we conclude that the probability of observing $0$ in the first register is $\frac{1}{2}\sum_{\sigma}p_{\sigma,0} + \frac{1}{2}\sum_{\sigma}p_{\sigma,1} = \frac{1}{2}$. A similar argument proves that we observe $1$ in the first register with probability exactly $\frac{1}{2}$.

\subsection{A New Protocol}\label{sec:new-protocol}

We will present a new quantum bit commitment protocol. We use the following quantum system between Alice (committer) and Bob (verifier): $\HH_{all} = \HH_{A,private} \otimes \HH_{bit} \otimes \HH_{open} \otimes \HH_{commit} \otimes \HH_{B, private}$, where $\HH_{A,private}$ is a system that is used only by Alice, $\HH_{open}$ holds a secret key produced by Alice, $\HH_{bit}$ is a 1-qubit system for a committed bit by Alice, $\HH_{commit}$ is used to produce an encrypted information regarding a committed bit, and $\HH_{B,private}$ is a system used only by Bob. Different from $\HH_{A,private}$ and $\HH_{B,private}$, the systems  $\HH_{bit}$, $\HH_{open}$, and $\HH_{commit}$ are accessed interchangeably by Alice and Bob at specific points during an execution of the protocol.

Consider the following bit commitment scheme between Alice and Bob. Let $n$ be the security parameter on which Alice and Bob initially agree.

We intend to include the description of the ownerships of each system that makes up $\HH_{all}$.  Moreover, we write $\HH_{open}$ for the system $\HH_{open1}\otimes \HH_{open2}$.

\s
\n\hrulefill \s\\
\n{\sc Committing Phase:}
\renewcommand{\labelitemi}{$\circ$}
\begin{itemize}\vs{-1}
  \setlength{\topsep}{-2mm}%
  \setlength{\itemsep}{1mm}%
  \setlength{\parskip}{0cm}%

\item[(C1)] Initially, Alice owns the system $\HH_{A}^{(C1)} = \HH_{A,private}\otimes  \HH_{bit} \otimes \HH_{open} \otimes \HH_{commit}$ and Bob owns $\HH_{B}^{(C1)} = \HH_{B,private}$. Starting with $\ket{0}$ in $\HH_{all}$, she randomly chooses her secret key $\pi\in K_n$ in $\HH_{open2}$.

\item[(C2)] She prepares $\ket{0}\ket{\mathrm{id}}$ in $\HH_{open1}\otimes \HH_{commit}$ and generates $\ket{\Phi_0^{(\pi)}}$ as described in Section \ref{sec:permutation}.

\item[(C3)] Let $a$ be a bit that Alice wants to commit. She produces $\ket{a}$ in $\HH_{bit}$. She then transforms $\ket{\Phi_{0}^{(\pi)}}$ in $\HH_{open1}\otimes \HH_{commit}$ into $\ket{\Phi_{a}^{(\pi)}}$ by applying $P_2$ when $a=1$.

\item[(C4)] She sends the system $\HH_{commit}$ to Bob.  Bob then receives the reduced state  $\rho_{a}^{(\pi)}$, which is called a {\em commitment state}. Bob should protect it from decoherence until the opening phase. In the end, Alice's system becomes $\HH_{A}^{(C4)} = \HH_{A,private}\otimes \HH_{bit}\otimes \HH_{open}$ and Bob's becomes  $\HH_{B}^{(C4)} = \HH_{commit}\otimes H_{B,private}$.
\end{itemize}\vs{-2}
\n\hrulefill\\

In the following opening phase, Alice reveals her secret bit $a$. Bob then checks whether it is actually the bit committed by her in the committing phase.

\n\hrulefill \s\\
\n{\sc Opening Phase (or Revealing Phase):}
\renewcommand{\labelitemi}{$\circ$}
\begin{itemize}\vs{-1}
  \setlength{\topsep}{-2mm}%
  \setlength{\itemsep}{1mm}%
  \setlength{\parskip}{0cm}%

\item[(R1)] Alice's current system is $\HH_{A}^{(R1)} = \HH_{A,private}\otimes \HH_{bit} \otimes \HH_{open}$ and Bob's system is $\HH_{B}^{(R1)} = \HH_{commit}\otimes \HH_{B,private}$.  Alice sends the system $\HH_{bit} \otimes \HH_{open}$ to Bob.

\item[(R2)] Alice now owns the system $\HH_{A}^{(R2)} = \HH_{A,private}$ and Bob owns $\HH_{B}^{(R2)} = \HH_{open}\otimes \HH_{bit}\otimes \HH_{commit} \otimes \HH_{B,private}$. Bob measures the two registers  $\HH_{bit} \otimes \HH_{open2}$ on the computational basis $\{0,1\}\otimes S_n$. Assume that he obtains $(a,\pi)$ after the measurement. If $\pi\not\in K_n$, then Bob declares that Alice tries to deceive him. In what follows, we assume that $\pi\in K_n$.

\item[(R3)] Assume that, in the previous committing phase, Bob had received a reduced state $\chi$ in $\HH_{commit}$ from Alice. Bob runs the state partition algorithm $C_{SPA}$ on input  $\density{0}{0}\otimes \chi$ in $\HH_{B,private}\otimes \HH_{commit}$, provided that $\HH_{B,private}$ is a $2$-dimensional system.

\item[(R4)] Measure the system $\HH_{B,private}$. If the outcome of the measurement is not $a$, then Bob declares that Alice tries to deceive him.

\item[(R5)] Whenever $a=1$, first apply $P_2$ to change $\ket{\Phi_{1}^{(\pi)}}$ to $\ket{\Phi_{0}^{(\pi)}}$. Apply $\tilde{P}_1^{-1}$ (which is given in  Section \ref{sec:permutation}) to $\HH_2\otimes \HH_{open}\otimes \HH_{commit}$.  Measure the system $\HH_{open1}\otimes \HH_{commit}$ in state $\ket{0}\ket{\mathrm{id}}$. If  $(0,\mathrm{id})$ is observed, then Bob accepts $a$ as Alice's committed bit. Otherwise, Bob declares that Alice tries to deceive him.
\end{itemize}\vs{-2}
\n\hrulefill\\

In Step (R2), Bob does not observe the subsystem $\HH_{opn1}$ because, otherwise, the entanglement between $\HH_{open1}$ and $\HH_{commit}$ could be destroyed and Steps (R3) and (R5) might not work properly.

In the subsequent section, we will analyze the above protocol in details.

\section{Security Analysis of the Scheme}\label{sec:analysis-protocol}

We will examine the security of the quantum bit commitment protocol given in Section \ref{sec:new-protocol}. We will show that our protocol is computationally concealing in Section \ref{sec:concealing-condition}. This is a direct consequence of \cite{KKNY12}. A more complex analysis is required to show the statistically binding condition in \ref{sec:binding-condition}.

\subsection{Computationally Concealing Condition}\label{sec:concealing-condition}

The concealing condition for bit commitment requires that Bob cannot retrieve any information on $a$ during a committing phase after honest Alice commits $a$ and sends a quantum state associated with $a$. Intuitively, this condition is satisfied
because Bob does not know $\gamma$, which locks the information on $a$ inside the quantum state, and thus there may be no way for Bob to obtain the information on $a$ with probability higher than a given parameter.


In our scheme, the notion of computational concealing, given in Section \ref{sec:main-contribution}, is rephrased as follows.  Our quantum bit commitment scheme is {\em computationally concealing} if, for any positive polynomial $p$, there is no polynomial-time quantum algorithm that outputs $a$ from $\rho_a$  with error probability at least $1/2+1/p(n)$ for any $n\in\nat^{+}$.


We will show that our protocol achieves the above computational concealing condition under the assumption that $\mathrm{GA}$ is hard to solve efficiently on a quantum computer.

\begin{theorem}\label{concealing-condition}
Let $n$ be an agreed security parameter. If no polynomial-time quantum algorithm solves $\mathrm{GA}$ with non-negligible probability, then our scheme satisfies the computational concealing condition.
\end{theorem}

Theorem \ref{concealing-condition} follows from the lemma below.

\begin{lemma}\label{distinguish-bound}
Let $k\in\nat^{+}$. If no polynomial-time quantum algorithm solves $\mathrm{GA}$ with error probability at least $2^{-n}$, where $n$ is the vertex set size of an input  graph, then Bob cannot distinguish between $\{\rho_{0}^{(\pi)}(n)^{\otimes k}\}_{n\in\nat^{+}}$ and $\{\rho_{1}^{(\pi)}(n)^{\otimes k}\}_{n\in\nat^{+}}$ with advantage at least $1/p(n)$ for any positive polynomial $p$.
\end{lemma}

Before proving Lemma \ref{distinguish-bound}, we give the proof of Theorem \ref{concealing-condition}.

\begin{proofof}{Theorem \ref{concealing-condition}}
Assume that there is an efficient quantum algorithm $\AAA$ that produces $a$ from $\rho_{a}^{(\pi)}$ with probability at least $1/2+1/p(n)$ for a certain fixed positive polynomial $p$. Let us consider the following quantum algorithm $\BB$: on input $\chi\in\{\rho_{0}^{(\pi)},\rho_{1}^{(\pi)}\}$, run $\AAA$ and obtain a bit, say, $a$.
By the property of $\AAA$, it follows that $\prob[\BB(1^n,\rho_{a}^{(\pi)})=a] \geq 1/2+1/p(n)$ for each bit $a\in\{0,1\}$. Thus, we obtain  $|\prob[\BB(1^n,\rho_{0}^{(\pi)})=1] - \prob[\BB(1^n,\rho_{1}^{(\pi)})=1]| \geq 2/p(n)$.  Hence, we can distinguish between  $\rho_0^{(\pi)}$ and $\rho_1^{(\pi)}$ with advantage at least $2/p(n)$. By Lemma \ref{distinguish-bound}, we conclude that $\mathrm{GA}$ is polynomial-time solvable on a quantum computer with non-negligible probability.
\end{proofof}

To prove Lemma \ref{distinguish-bound}, we recall the computational distinction problem $\mathrm{QSCD_{ff}}$ introduced by Kawachi \etalc~\cite{KKNY12}. Since we need only a restricted form of this problem, we re-formulate this problem in the following fashion. Let $k$ be a fixed constant in $\nat^{+}$.

\ms
\n\hs{3}{\sc $k$-Quantum State Computational Distinction Problem} $k\mbox{-}\mathrm{QSCD_{ff}}$ (weaker version):
\renewcommand{\labelitemi}{$\circ$}
\begin{itemize}\vs{-1}
  \setlength{\topsep}{-2mm}%
  \setlength{\itemsep}{1mm}%
  \setlength{\parskip}{0cm}%

\item {\sc Instance:} a string $1^n$ and a $k$-fold quantum state $\rho^{\otimes k}$ with $\rho\in \{ \rho_{0}^{(\pi)}(n), \rho_{1}^{(\pi)}(n) \}$ for a certain fixed (but hidden) permutation $\pi\in K_n$, depending only on $n$, where $n\in\nat^{+}$.

\item {\sc Output:} YES, if $\rho = \rho_{0}^{(\pi)}(n)$; NO, otherwise.
\end{itemize}

For convenience, we say that a quantum algorithm $\AAA$ {\em solves $k\mbox{-}\mathrm{QSCD_{ff}}$ with advantage} $p(n)$ if $\AAA$ distinguishes between $\{\rho_{0}^{(\pi)}(n)^{\otimes k}\}_{n\in\nat^{+}}$ and $\{\rho_{1}^{(\pi)}(n)^{\otimes k}\}_{n\in\nat^{+}}$ with advantage $p(n)$.
Moreover, we say that a quantum algorithm $\AAA$ solves $k\mbox{-}\mathrm{QSCD_{ff}}$ with {\em average advantage} $p$ on length $n$
if the average, over all $\pi\in K_n$ chosen uniformly at random, of the advantage with which $\AAA$ distinguishes between $\{\rho_{0}^{(\pi)}(n)^{\otimes k}\}_{n\in\nat^{+}}$ and $\{\rho_{1}^{(\pi)}(n)^{\otimes k}\}_{n\in\nat^{+}}$ is exactly $p$.
We note that, by
combining \cite[Theorem 2.2]{KKNY12} and \cite[Theorem 2.5]{KKNY12}, the following holds.

\begin{lemma}\label{GA-solvable}{\rm \cite{KKNY12}}
Let $k\in\nat^{+}$. If a quantum algorithm $\AAA$ solves $k\mbox{-}\mathrm{QSCD_{ff}}$ with average advantage at least $1/p(n)$ for a certain positive polynomial $p$, then there exists a quantum algorithm that solves $\mathrm{GA}$ for infinitely-many lengths with probability at least $1-2e^{-n}$, where $e$ is the base of natural logarithms and $n$ refers to the vertex set size of an input graph of $\mathrm{GA}$.
\end{lemma}

At last, we return to the proof of Lemma \ref{distinguish-bound}.

\begin{proofof}{Lemma \ref{distinguish-bound}}
We will show the contrapositive of the lemma. Let $k\in\nat^{+}$.
Assume that there are a positive polynomial $p$ and a polynomial-time quantum algorithm $\AAA$ that distinguishes between $\rho_{0}^{(\pi)}(n)^{\otimes k}$ and $\rho_{0}^{(\pi)}(n)^{\otimes k}$ with advantage at least $1/p(n)$.
In other words, we can solve $k\mbox{-}\mathrm{QSCD_{ff}}$ in polynomial time with advantage at least $1/p(n)$. Lemma \ref{GA-solvable} therefore implies that $\mathrm{GA}$ is solvable for infinitely-many input lengths $n$ on an appropriate quantum computer in polynomial time with error probability at most $2e^{-n}$, which is bounded from above by $2^{-n}$.
\end{proofof}

\subsection{Statistically Binding Condition}\label{sec:binding-condition}

The binding condition for classical bit commitment requires that adversarial Alice cannot cheat Bob simply by revealing a different bit $a'$ together with a different key $\pi'$ to Bob. For quantum bit commitment,
Dumais \etalc~\cite{DMS00} formulated a condition for a quantum bit commitment scheme to be {\em statistically binding} in the case of non-interactive schemes. Other definitions for binding condition are found in, \eg \cite{DFR+07}.

Conventionally, we say that Alice {\em unveils} $a$ (with probability $p$) if, in the opening phase,  Bob observes $a$ and convinces himself that this is truly a committed bit (with probability $p$).

Here, we cope with a general adversary model, proposed in \cite{DMS00}, which describes adversarial Alice's attack $\UU$ as a triplet $(U_1,U_2^{(0)},U_2^{(1)})$ of unitary operators.

\ms

\sloppy
(1) At the beginning of the committing phase, adversarial Alice starts with the initial state $\ket{0}$ in her system $\HH_{A}^{(C1)} = \HH_{A,private}\otimes \HH_{bit}\otimes \HH_{open} \otimes \HH_{commit}$.
Instead of taking Steps (C1)--(C3), she applies  the unitary operator $U_1$ to $\ket{0}$  in $\HH_{A}^{(C1)}$ and generates a quantum state  $\ket{\eta^{(C1)}} = U_1\ket{0} = \sum_{a\in\{0,1\}} \sum_{\sigma\in S_n} \sum_{\pi\in K_n} \ket{\xi_{a,\pi,\sigma}}\ket{a} \ket{\pi}\ket{\sigma} \ket{\gamma_{a,\pi,\sigma}}$ with $\sum_{a,\sigma,\pi}\|\ket{\xi_{a,\pi,\sigma}}\|^2=1$.
Since $\ket{\gamma_{a,\pi,\sigma}}\in \HH_{S_n}=\mathrm{span}\{\ket{\phi_{\tau,s}^{(\pi)}}\mid \tau\in S_n^{(\pi)},s\in\{0,1\}\}$, it holds that $\ket{\gamma_{a,\pi,\sigma}} = \sum_{\tau,s}\alpha_{\tau,s}^{(a,\pi,\sigma)}
\ket{\phi_{\tau,s}^{(\pi)}}$ for an appropriate set  $\{\alpha_{\tau,s}^{(a,\pi,\sigma)}\}_{\tau,s}$.
Hence, we obtain
\begin{equation}\label{eqn:state-in-C1}
\ket{\eta^{(C1)}} = \sum_{\pi\in K_n}\sum_{\sigma,\tau\in S_{n}^{(\pi)}}\sum_{a,s\in\{0,1\}} \ket{\xi_{a,\pi,\sigma}^{(\tau,s)}} \ket{a} \ket{\pi} \ket{\sigma} \ket{\phi_{\tau,s}^{(\pi)}},
\end{equation}
where $\ket{\xi_{a,\pi,\sigma}^{(\tau,s)}} = \alpha_{\tau,s}^{(a,\pi,\sigma)}\ket{\xi_{a,\pi,\sigma}}$.
At Step (C4), she sends the system $\HH_{commit}$ to Bob.
The quantum state that Bob receives from adversarial Alice is
of the form  $\chi=\sum_{a,s,\sigma,\tau,\pi}
\|\ket{\xi_{a,\pi,\sigma}^{(\tau,s)}}\|^2  \density{\phi_{\tau,s}^{(\pi)}}{\phi_{\tau,s}^{(\pi)}}$ instead of $\frac{1}{|S_n|}\sum_{\sigma}\density{\phi_{\sigma,a}^{(\pi)}}{\phi_{\sigma,a}^{(\pi)}}$.

(2) At the beginning of the opening phase, recall that Alice owns the subsystem  $\HH_{A}^{(R1)} = \HH_{A,private}\otimes \HH_{bit}\otimes \HH_{open}$. Before sending $\HH_{bit}\otimes \HH_{open}$ to Bob, Alice applies one of the unitary operators, $U_2^{(0)}$ and $U_2^{(1)}$, in an attempt to unveil $0$ and $1$, respectively.
Assume that Alice chooses a bit $a\in\{0,1\}$ and Alice tries to apply  $U_2^{(a)}$ in order to maximize the probability of unveiling $a$.
Now, she applies $U_2^{(a)}\otimes I$ to $\ket{\eta^{(C1)}}$, where $U_2^{(a)}$ acts on $\HH_{A}^{(R1)}$ and $I$ acts on $\HH_{commit}$, and then obtains
$
\ket{\eta^{(R1)}} = (U_2^{(a)}\times I)\ket{\eta^{(C1)}}
$
in $\HH_{A}^{(R1)}\otimes \HH_{commit}$.
Given Alice's cheating strategy $\UU=(U_1,U_2^{(0)},U_2^{(1)})$, for each bit $a$, recall that $T_a^{(\UU)}(n)$ denotes the probability that, when Alice applies $(U_1,U_2^{(a)})$ as described above, Bob obtains $a$ by the projection measurement in Step (R2) (with ignoring the values of $\pi$ and $\sigma$) and then accepts Alice's bit $a$ through Steps (R3)--(R5) (in other words, Alice successfully unveils $a$). Note that $0\leq S_0(n)+S_1(n)\leq 2$.

The opening phase of our scheme, assuming that Bob honestly follows Steps (R2)--(R5), is in essence equivalent to the process that Bob immediately measures the system $\HH_{bit}\otimes \HH_{open}\otimes \HH_{commit}$ in state  $\ket{a}\ket{\pi}\ket{\Phi_{a}^{(\pi)}}$ for any $\pi\in K_n$.
To be more precise, we need to introduce four measurement operators $M_{bit}^{(a)}$, $M_{open1}^{(\pi)}$, $M_{open2}^{(\sigma)}$, and $M_{commit}^{(a,\pi,\sigma)}$ acting on $\HH_{bit}$, $\HH_{open1}$, $\HH_{open2}$, and $\HH_{commit}$ that project onto states $\ket{a}$,  $\ket{\pi}$, $\ket{\sigma}$, and $\ket{\phi_{\sigma,a}^{(\pi)}}$, respectively. For convenience, we set  $M_{open}^{(\pi,\sigma)}\equiv M_{open1}^{(\pi)}\otimes M_{open2}^{(\sigma)}$. Let  $M_{mix}^{(a,\pi)}$ be a measurement operator acting on $\HH_{open2}\otimes \HH_{commit}$ projecting onto state $\ket{\Phi_{a}^{(\pi)}}$. Finally, we define $M_a \equiv \sum_{\pi\in K_n} M_{bit}^{(a)}\otimes M_{open1}^{(\pi)}\otimes M_{mix}^{(a,\pi)}$ for each $a\in\{0,1\}$.
In our argument that follows shortly, we assume that, instead of Steps (R2)--(R5), Bob simply performs $M_a$ in the system $\HH_{bit}\otimes \HH_{open}\otimes \HH_{commit}$.  Note that
the value $T_a^{(\UU)}(n)$ associated with Alice's cheating strategy $\UU$  coincides with $\| (I\otimes M_s)(U_2^{(s)}\otimes I)\ket{\eta^{(C1)}} \|^2$.


Let us recall from Section \ref{sec:main-contribution} the notion of statistical binding.
We rephrase this notion as follows. Our quantum bit commitment scheme is
{\em statistically binding} if there exists a negligible function $\varepsilon(n)$ such that, for any cheating strategy  $\UU=(U_1,U_2^{(0)},U_2^{(1)})$ of Alice, $T_0^{(\UU)}(n)+T_1^{(\UU)}(n)\leq 1+  \varepsilon(n)$ holds for every length $n\in\nat^{+}$.


We will show that our protocol achieves the statistical binding condition.

\begin{theorem}\label{binding-result}
If no polynomial-time quantum algorithm solves $\mathrm{GA}$ with non-negligible probability, then our quantum bit commitment scheme is statistically binding.
\end{theorem}

To prove this theorem, we first introduce a new problem, called the {\em hidden permutation search problem}, which is closely related to the indistinguishability between $\rho_0^{(\pi)}$ and $\rho_1^{(\pi)}$.

\ms
\n\hs{3}{\sc Hidden Permutation Search Problem} $\mathrm{HPSP}$:
\renewcommand{\labelitemi}{$\circ$}
\begin{itemize}\vs{-1}
  \setlength{\topsep}{-2mm}%
  \setlength{\itemsep}{1mm}%
  \setlength{\parskip}{0cm}%

\item {\sc Instance:} a string $1^n$ with $n\in\nat^{+}$ and a quantum state $\rho_{0}^{(\pi)}(n)$ with a certain hidden permutation $\pi\in K_n$.

\item {\sc Output:} $\pi$.
\end{itemize}

We say that a quantum algorithm $\AAA$ {\em solves $\mathrm{HPSP}$ with average probability} $p$ on length $n$ if, over all permutations $\pi\in K_n$ chosen uniformly at random, $\AAA$ takes instance $(1^n,\rho_{0}^{(\pi)}(n))$ and outputs $\pi$ with probability exactly $p$.

\begin{lemma}\label{HPSP-S0-S1-bound}
Assume that there exist a cheating strategy $\UU$ of Alice and a positive polynomial $p$ satisfying $T_0^{(\UU)}(n)+T_1^{(\UU)}(n)\geq 1+1/p(n)$ for infinitely-many lengths $n$. Then, there exist a positive polynomial $q$ and a polynomial-time quantum algorithm that solves $\mathrm{HPSP}$ with average probability at least $1/q(n)$ for infinitely-many lengths $n$.
\end{lemma}

Because the proof of Lemma \ref{HPSP-S0-S1-bound} requires a special treatment, we will give it in the next section.

\begin{lemma}\label{HSP-vs-QSCDff}
If there are a positive polynomial $p$ and a polynomial-time quantum algorithm that solves $\mathrm{HPSP}$ with average probability at least $1/p(n)$ for infinitely-many $n$, then there are a positive polynomial $q$ and a polynomial-time quantum algorithm that solves $2\mbox{-}\mathrm{QSCD_{ff}}$ with average advantage at least $1/q(n)$ for infinitely-many lengths $n$.
\end{lemma}

\begin{proof}
Let $p$ be a positive polynomial and let $\AAA$ be a polynomial-time quantum algorithm that solves $\mathrm{HPSP}$ with probability at least $1/p(n)$. Let $\rho\otimes \rho$ with $\rho\in\{\rho_0^{(\pi)},\rho_1^{(\pi)}\}$ be an instance of $2\mbox{-}\mathrm{QSCD_{ff}}$, where $\pi$ is an unknown permutation in $K_n$. Let $s\in\{0,1\}$ and assume that  $\rho=\rho_s^{(\pi)}$. Now, our task is to determine whether $s=0$ or $s=1$.

Let $\AAA_0$ be a polynomial-time quantum algorithm that generates $\pi$ from input $(1^n,\rho_0^{(\pi)})$ with probability, say, $\gamma_n$, which is at least $1/p(n)$. Let $\AAA_1$ be a quantum algorithm that takes input $\rho_1^{(\pi)}$, transforms $\rho_1^{(\pi)}$ to $\rho_0^{(\pi)}$ by running $P_2$, and finally apply $\AAA_0$. Obviously, $\AAA_1$ outputs $\pi$ with the same probability as $\AAA_0$ with $\rho_0^{(\pi)}$. Now, let us consider input $\rho_s^{(\pi)}\otimes \rho_s^{(\pi)}$ with unknown values $s$ and $\pi$.

Using the first state $\rho=\rho_s^{(\pi)}$, we obtain $\pi$ as follows.
Consider a purification $\ket{\Phi}= \ket{\Phi_{s}^{(\pi)}}$ of $\rho$. Starting with $\ket{0}\ket{\Phi}$, apply $H\otimes I$ to obtain $\frac{1}{\sqrt{2}} (\ket{0}\ket{\Phi}+\ket{1}\ket{\Phi})$. We apply $\AAA_0$ and $\AAA_1$ separately to generate $\frac{1}{\sqrt{2}} (\ket{0}\otimes (I\otimes \AAA_0)\ket{\Phi}+ \ket{1}\otimes (I\otimes \AAA_1)\ket{\Phi})$. It follows that $\sum_{\kappa\in K_n} \ket{0} \ket{\kappa}\ket{\xi_{\kappa,0}}  + \sum_{\kappa\in K_n}\ket{1}  \ket{\kappa}\ket{\xi_{\kappa,1}}$. By our assumption, we obtain $\|\ket{\xi_{\pi,s}}\|^2 = \gamma_n$ for every $s\in\{0,1\}$.

Next, we apply $C_{SPA}$ to the second register and the second input state $\rho_{s}^{(\pi)}$; that is, $\density{\kappa}{\kappa}\otimes \rho_{s}^{(\pi)}$. If $\kappa=\pi$, this process produces $\density{s}{s}\otimes \rho_{s}^{(\pi)}$. Finally, we observe the first register and output its content $s$.
If $\kappa$ is different from $\pi$, then
after running $C_{SPA}$, we observe $0$ and $1$ with equal probability, as we have argued in Section \ref{sec:partition}.

Therefore, the probability that we correctly obtain $s$ is exactly $\frac{1}{2}(1-\gamma_n) + \gamma_n = \frac{1}{2}+\frac{\gamma_n}{2}$. Calling the entire quantum algorithm by $\BB$, we have just proven that  $\prob[\BB(1^n,\rho_s^{(\pi)} \otimes \rho_s^{(\pi)}) = s] = \frac{1}{2}+\frac{\gamma_n}{2}$. From this equation, it follows that
\[
\left|\prob[\BB(1^n,\rho_0^{(\pi)} \otimes \rho_0^{(\pi)}) =1] - \prob[\BB(1^n,\rho_1^{(\pi)} \otimes \rho_1^{(\pi)}) =1]\right|  = \left| \left( \frac{1}{2} -\frac{\gamma_n}{2} \right) - \left( \frac{1}{2} +\frac{\gamma_n}{2} \right) \right| = \gamma_n.
\]
Since $\gamma_n\geq 1/p(n)$, the lemma follows.
\end{proof}

Using the above lemmas, we can prove Theorem \ref{binding-result}.

\begin{proofof}{Theorem \ref{binding-result}}
We want to show the contrapositive of the theorem. First, assume that there exist a positive polynomial $p$ and Alice's cheating strategy $\UU=(U_1,U_2^{(0)},I)$ such that $T_0^{(\UU)}(n)+T_1^{(\UU)}(n)\geq 1+1/p(n)$ for infinitely-many lengths $n$. By Lemmas \ref{HPSP-S0-S1-bound} and \ref{HSP-vs-QSCDff}, we conclude that $2\mbox{-}\mathrm{QSCD_{ff}}$ can be solved by a certain polynomial-time quantum algorithm with average advantage at least $1/q(n)$ for a certain polynomial $q$. By Lemma \ref{GA-solvable}, GA must be solved on a quantum computer in polynomial time with average probability at least $1/r(n)$ for a certain polynomial $r$.
\end{proofof}

\section{Quantum Algorithm for HPSP}\label{sec:algorithm-HPSP}

In the previous section, we have left Lemma \ref{HPSP-S0-S1-bound} unproven. Here, we will give its missing proof by constructing an appropriate quantum algorithm that solves HPSP with non-negligible probability, provided that the statistically binding condition does not
hold.

Recall that adversarial Alice takes $(U_1,U_2^{(0)},U_2^{(1)})$ as her cheating strategy. To simplify our analysis, as in \cite{DMS00}, we replace $(U_1,U_2^{(0)},U_2^{(1)})$ by $(\tilde{U}_1,\tilde{U}_2^{(0)},I)$, where $\tilde{U}_1 = (U_2^{(1)}\otimes I_{commit})U_1$ and $\tilde{U}_2^{(0)} = U_2^{(0)}(U_2^{(1)})^{\dagger}$, without changing the probability that Alice successfully cheats Bob. For convenience, hereafter, we write $U_1$ and $U_2^{(0)}$ (without ``tilde'') for $\tilde{U}_1$ and $\tilde{U}_2^{(0)}$, respectively, and deal only with $\UU=(U_1,U_2^{(0)},I)$ as adversarial Alice's cheating strategy, where $U_2^{(1)}=I$.

\subsection{Distillation Algorithm}\label{sec:distilation}

We will present an important subroutine that makes up of the quantum algorithm that solves HPSP in Section \ref{sec:HPSP-algorithm}. Recall that adversarial Alice is now taking the cheating  strategy $\UU=(U_1,U_2^{(0)},I)$, while Bob faithfully follows the protocol.

Let us recall from Eq.~(\ref{eqn:state-in-C1}) that $\ket{\eta^{(C1)}}$ is
of the form $\sum_{\pi\in K_n}\sum_{\sigma,\tau\in S_{n}^{(\pi)}}\sum_{a,s\in\{0,1\}} \ket{\xi_{a,\pi,\sigma}^{(\tau,s)}} \ket{a} \ket{\pi} \ket{\sigma} \ket{\phi_{\tau,s}^{(\pi)}}$, which is obtained by an application of $U_1\otimes I$ to $\ket{0}$ in the entire system $\HH_{all} = \HH_{A}^{(C1)}\otimes \HH_{B}^{(C1)}$.
Note that $T_1^{(\UU)}(n) = \| (I\otimes M_1)\ket{\eta^{(C1)}} \|^2$ since $U_2^{(1)}=I$.
For convenience, we define $\ket{\eta^{(C1)}_{perf}}$ to be the normalized state of $(I\otimes M_1)\ket{\eta^{(C1)}}$; that is,  $\ket{\eta^{(C1)}_{perf}} = \frac{1}{\sqrt{T_1^{(\UU)}(n)}} (I\otimes M_1)\ket{\eta^{(C1)}}$. This is
an ideal quantum state for adversarial Alice because,  from this state,
Alice unveils $1$ with certainty.
By the definition of $M_1$, we can assume that
$\ket{\eta^{(C1)}_{perf}}$ has the form $\sum_{\pi\in K_n} \ket{\xi_{1,\pi}} \ket{1}\ket{\pi}\ket{\Phi_{1}^{(\pi)}}$ with $\sum_{\pi\in K_n}\| \ket{\xi_{1,\pi}} \|^2=1$.

Now, we will demonstrate how to implement $I\otimes M_1$ algorithmically and distill $\ket{\eta^{(C1)}_{perf}}$ from $\ket{\eta^{(C1)}}$ using the measurement in the computational basis.

\s
\n\hrulefill \s\\
\n{\sc Distillation Algorithm $\AAA_{dis}$:}
\renewcommand{\labelitemi}{$\circ$}
\begin{itemize}\vs{-1}
  \setlength{\topsep}{-2mm}%
  \setlength{\itemsep}{1mm}%
  \setlength{\parskip}{0cm}%

\item[(D1)] Prepare an additional register in state $\ket{0}$ in $\HH_2$. Given $\ket{0}\ket{\eta^{(C1)}}$ in $\HH_2\otimes \HH_{A}^{(C1)}$,
    we focus on a term $\ket{\xi_{a,\pi,\sigma}^{(\tau,s)}}\ket{a}\ket{\pi} \ket{\sigma} \ket{\phi_{\tau,s}^{(\pi)}}$ in the subsystem $\HH_{A}^{(C1)}$. We measure the second register in state $\ket{1}$. The state collapses to $\ket{\xi_{1,\pi,\sigma}^{(\tau,s)}}\ket{1}\ket{\pi}\ket{\sigma} \ket{\phi_{\tau,s}^{(\pi)}}$.

\item[(D2)] Transform $\ket{0}\ket{\pi}\ket{\phi_{\tau,s}^{(\pi)}}$ in $\HH_2\otimes \HH_{open2}\otimes \HH_{commit}$ to $\ket{s}\ket{\pi}\ket{\phi_{\tau,s}^{(\pi)}}$ by applying $P_{SPA}$ given in Section \ref{sec:partition}. Measure the first register in state $\ket{1}$. We then obtain $\ket{\xi_{1,\pi,\sigma}^{(\tau,1)}}\ket{1}\ket{\pi}\ket{\phi_{\tau,1}^{(\pi)}}$.

\item[(D3)] Change $\ket{\phi_{\tau,1}^{(\pi)}}$ to $\ket{\phi_{\tau,0}^{(\pi)}}$ by applying $P_2$. Prepare $\ket{\mathrm{id}}$ and  transform $\ket{\pi}\ket{\mathrm{id}} \ket{\phi_{\tau,1}^{(\pi)}}$ in $\HH_{open1}\otimes \HH_{S_n}\otimes \HH_{commit}$ to $\ket{\pi}\ket{\tau} \ket{\mathrm{id}}$ by applying $P_1^{-1}$. We then obtain $\ket{\xi_{1,\pi,\sigma}^{(\tau,1)}}\ket{\pi} \ket{\tau} \ket{\mathrm{id}}$. Measure the fourth register in state $\ket{\sigma}$ to obtain $\ket{\xi_{1,\pi,\sigma}^{(\sigma,1)}}\ket{\pi} \ket{\sigma}\ket{\mathrm{id}}$.

\item[(D4)] Apply $P_1$ to $\ket{\sigma}\ket{\pi}\ket{\mathrm{id}}$ in $\HH_{open}\otimes \HH_{commit}$ to obtain $\ket{\pi}\ket{\sigma}\ket{\phi_{\sigma,1}^{(\pi)}}$.
\end{itemize}\vs{-2}
\n\hrulefill\\

It is not difficult to see that the above algorithm $\AAA_{dis}$ transforms $\ket{\eta^{(C1)}}$ into a quantum state $\ket{\eta^{(C1)}_{perf}}$ with probability $T_1^{(\UU)}(n) = \|(I\otimes M_1)\ket{\eta^{(C1)}} \|^2$.

In the following argument, we are focused on $\ket{\eta^{(C1)}_{perf}}$. Now,  we fix a permutation $\pi'\in K_n$, which is a hidden permutation of an instance $\rho_{0}^{(\pi)}$ of HPSP.
Note that the Hilbert space $span\{\ket{\phi_{\sigma,0}^{(\pi')}}\mid \sigma\in S_n\}$ is determined by a basis $\BB_{0}^{(\pi')} = \{\ket{\phi_{\sigma,0}^{(\pi')}}\mid \sigma\in S_{n}^{(\pi')}\}$, where $S_n^{(\pi')}$ is defined in Section \ref{sec:permutation}.
First, we want to measure  $\HH_{commit}$ in states $\ket{\phi_{\sigma,0}^{(\pi')}}$ for an arbitrary permutation  $\sigma\in S_n^{(\pi')}$.
This is formally done by
a measurement operator $\tilde{M}_{\pi'} \equiv \sum_{\sigma\in S_n^{(\pi')}}M_{commit}^{(0,\sigma,\pi')}$, which projects a quantum state in   $\HH_{commit}$ onto $\ket{\phi_{\sigma,0}^{(\pi')}}$'s.
Letting $\ket{\eta^{(\pi')}} = (I\otimes \tilde{M}_{\pi'})\ket{\eta^{(C1)}_{perf}}$, we want to determine an actual form of $\ket{\eta^{(\pi')}}$. For brevity, we set $\omega_{n} = \frac{|K_n|+1}{\sqrt{2|S_n|}(|K_n|-1)}$.

\begin{lemma}
For each fixed $\pi'\in K_n$,  $\ket{\eta^{(\pi')}} =
\omega_{n} \sum_{\sigma\in S_n}\sum_{\pi\in K_n} \ket{\xi_{1,\pi}}\ket{1}\ket{\pi} \ket{\phi_{\sigma,1}^{(\pi)}}\ket{\phi_{\sigma,0}^{(\pi')}}$.
\end{lemma}

\begin{proof}
As the first step, we intend to  express $\ket{\eta^{(C1)}_{perf}}$ in terms of $\ket{\phi_{\sigma,0}^{(\pi')}}$.
Recall that  $\ket{\eta^{(C1)}_{perf}} = \frac{1}{\sqrt{|S_n|}} \sum_{\sigma\in S_n} \sum_{\pi\in K_n} \ket{\xi_{1,\pi}} \ket{1}\ket{\pi} \ket{\sigma} \ket{\phi_{\sigma,1}^{(\pi)}}$.
For convenience, let  $\ket{\Theta_{\pi,\sigma}} =  \frac{1}{\sqrt{|S_n|}} \ket{\xi_{1,\pi}} \ket{1}\ket{\pi}\ket{\sigma}$ and
$\delta = \frac{1}{|K_n|-1}$. Since $\ket{\eta^{(C1)}_{perf}} = \sum_{\pi\in K_n} \ket{\xi_{1,\pi}}\ket{1} \ket{\pi} \ket{\Phi_{1}^{(\pi)}}$,  $\ket{\eta^{(C1)}_{perf}}$ can be expressed as $\sum_{\pi,\sigma}\ket{\Theta_{\pi,\sigma}}\ket{\phi_{\sigma,1}^{(\pi)}}$.
Since $\ket{\phi_{\sigma,1}^{(\pi)}} =  \delta \sum_{\kappa\in K_n} (\ket{\phi_{\sigma,0}^{(\kappa)}} - \ket{\phi_{\sigma\pi,0}^{(\kappa)}})$
by Lemma \ref{Phi-rewritten},  $\ket{\eta^{(C1)}_{perf}}$ is written as
$\delta \sum_{\kappa\in K_n} \sum_{\pi,\sigma}\ket{\Theta_{\pi,\sigma}} \ket{\phi_{\sigma,0}^{(\kappa)}}  - \delta \sum_{\kappa\in K_n} \sum_{\pi,\sigma}\ket{\Theta_{\pi,\sigma}} \ket{\phi_{\sigma\pi,0}^{(\kappa)}}$.
Since $\sum_{\pi,\sigma}\ket{\Theta_{\pi,\sigma}} \ket{\phi_{\sigma\pi,0}^{(\kappa)}}$ equals
$\sum_{\pi,\sigma}\ket{\Theta_{\pi,\sigma\pi}} \ket{\phi_{\sigma,0}^{(\kappa)}}$,  the state
$\ket{\eta^{(C1)}_{perf}}$ is further written as
\begin{equation}\label{eqn:eta-C1-rewritten-pi}
\ket{\eta^{(C1)}_{perf}} = \delta \sum_{\kappa\in K_n} \sum_{\sigma\in S_n}\sum_{\pi\in K_n}(\ket{\Theta_{\pi,\sigma}} - \ket{\Theta_{\pi,\sigma\pi}}) \ket{\phi_{\sigma,0}^{(\kappa)}}
= \frac{\delta\sqrt{2}}{\sqrt{|S_n|}} \sum_{\kappa} \sum_{\sigma}\sum_{\pi} \ket{\xi_{1,\pi}}\ket{1}\ket{\pi} \ket{\phi_{\sigma,1}^{(\pi)}}\ket{\phi_{\sigma,0}^{(\kappa)}},
\end{equation}
since $\ket{\Theta_{\pi,\sigma}} - \ket{\Theta_{\pi,\sigma\pi}} = \frac{1}{\sqrt{|S_n|}} \ket{\xi_{1,\pi}}\ket{1}\ket{\pi} (\ket{\sigma}-\ket{\sigma\pi}) = \frac{\sqrt{2}}{\sqrt{|S_n|}} \ket{\xi_{1,\pi}}\ket{1}\ket{\pi} \ket{\phi_{\sigma,1}^{(\pi)}}$.

Obviously, if $\sigma\in S_n^{(\pi')}$, then $\tilde{M}_{\pi'}\ket{\phi_{\sigma,0}^{(\pi')}} = \ket{\phi_{\sigma,0}^{(\pi')}}$ holds. When $\sigma\in S_n-S_n^{(\pi')}$, we obtain $\tilde{M}_{\pi'}\ket{\phi_{\sigma\pi',0}^{(\pi')}} = \tilde{M}_{\pi'}\ket{\phi_{\sigma,0}^{(\pi')}} = \ket{\phi_{\sigma,0}^{(\pi')}} = \ket{\phi_{\sigma\pi',0}^{(\pi')}}$. Let  $\kappa\in K_n-\{\pi'\}$. For any $\sigma\in S_n^{(\pi')}$, since  $\measure{\phi_{\sigma,0}^{(\pi')}}{\phi_{\sigma,0}^{(\kappa)}} = \frac{1}{2}$ by Lemma \ref{base-phi-1}(4),
it follows that $\tilde{M}_{\pi'} \ket{\phi_{\sigma,0}^{(\kappa)}} = \frac{1}{2}\ket{\phi_{\sigma,0}^{(\pi')}}$. Moreover, since Lemma \ref{base-phi-1}(4) implies   $\measure{\phi_{\sigma,0}^{(\pi')}}{\phi_{\sigma\pi',0}^{(\kappa)}} = \frac{1}{2}$, we obtain $\tilde{M}_{\pi'} \ket{\phi_{\sigma\pi',0}^{(\kappa)}} = \frac{1}{2}\ket{\phi_{\sigma,0}^{(\pi')}}$, which equals $\frac{1}{2}\ket{\phi_{\sigma\pi',0}^{(\pi')}}$. Overall, it holds that $\tilde{M}_{\pi'} \ket{\phi_{\sigma,0}^{(\kappa)}} = \frac{1}{2}\ket{\phi_{\sigma,0}^{(\pi')}}$ for all permutations $\sigma\in S_n$.

Since $\ket{\eta^{(\pi')}} = (I\otimes \tilde{M}_{\pi'})\ket{\eta^{(C1)}_{perf}}$, from Eq.~(\ref{eqn:eta-C1-rewritten-pi}), $\ket{\eta^{(\pi')}}$ is expressed as
\begin{eqnarray*}\label{eqn:eta-C1-rephrased}
\ket{\eta^{(\pi')}} &=& \frac{\delta\sqrt{2}}{\sqrt{|S_n|}} \left[
\sum_{\sigma}\sum_{\pi} \ket{\xi_{1,\pi}}\ket{1}\ket{\pi} \ket{\phi_{\sigma,1}^{(\pi)}}\ket{\phi_{\sigma,0}^{(\pi')}}
+ \frac{1}{2}
\sum_{\sigma}\sum_{\pi} \sum_{\kappa:\kappa\neq \pi'} \ket{\xi_{1,\pi}}\ket{1}\ket{\pi} \ket{\phi_{\sigma,1}^{(\pi)}}\ket{\phi_{\sigma,0}^{(\pi')}} \right] \\
&=& \frac{\delta\sqrt{2}}{\sqrt{|S_n|}}  \left[
\sum_{\sigma}\sum_{\pi} \ket{\xi_{1,\pi}}\ket{1}\ket{\pi} \ket{\phi_{\sigma,1}^{(\pi)}}\ket{\phi_{\sigma,0}^{(\pi')}}
+ \frac{|K_n|-1}{2}
\sum_{\sigma}\sum_{\pi}  \ket{\xi_{1,\pi}}\ket{1}\ket{\pi} \ket{\phi_{\sigma,1}^{(\pi)}}\ket{\phi_{\sigma,0}^{(\pi')}} \right] \\
&=& \frac{\delta\sqrt{2}}{\sqrt{|S_n|}}\cdot \frac{|K_n|+1}{2}
\sum_{\sigma}\sum_{\pi} \ket{\xi_{1,\pi}}\ket{1}\ket{\pi} \ket{\phi_{\sigma,1}^{(\pi)}}\ket{\phi_{\sigma,0}^{(\pi')}}.
\end{eqnarray*}
Therefore, we obtain
$\ket{\eta^{(\pi')}} = \frac{|K_n|+1}{\sqrt{2|S_n|}(|K_n|-1)}  \sum_{\sigma}\sum_{\pi} \ket{\xi_{1,\pi}}\ket{1}\ket{\pi} \ket{\phi_{\sigma,1}^{(\pi)}}\ket{\phi_{\sigma,0}^{(\pi')}}$ by the definition of $\delta$.
\end{proof}

For convenience, we denote by $\ket{\eta^{(\pi')}_{norm}}$ the normalized state of $\ket{\eta^{(\pi')}}$, \ie $\ket{\eta^{(\pi')}_{norm}} = \frac{1}{\|\ket{\eta^{(\pi')}}\|}\ket{\eta^{(\pi')}}$.

\begin{lemma}\label{norm-of-eta-pi}
$\ket{\eta^{(\pi')}_{norm}} = \omega'_n \sum_{\sigma\in S_n}\sum_{\pi\in K_n} \ket{\xi_{1,\pi}} \ket{1} \ket{\pi} \ket{\phi_{\sigma,1}^{(\pi)}} \ket{\phi_{\sigma,0}^{(\pi')}}$, where $\omega'_n = \frac{1}{\sqrt{|S_n|(1-\|\ket{\xi_{1,\pi'}}\|^2)}}$.
\end{lemma}

\begin{proof}
We want to estimate the value $\|\ket{\eta^{(\pi')}}\|$.
First, we claim that (*) $\|\ket{\eta^{(\pi')}}\|^2 = \frac{(1-\|\ket{\xi_{1,\pi'}}\|^2)(|K_n|+1)^2}{2(|K_n|-1)^2}$. If this claim is true, then $\ket{\eta^{(\pi')}_{norm}}$ is written as
\begin{eqnarray*}
\ket{\eta^{(\pi')}_{norm}} &=&
\left[ \left( \frac{|K_n|+1}{\sqrt{2|S_n|}(|K_n|-1)} \right)\left/\right. \left( \frac{\sqrt{1-\|\ket{\xi_{1,\pi'}}\|^2} (|K_n|+1)}{\sqrt{2}(|K_n|-1)} \right)  \right] \sum_{\sigma}\sum_{\pi} \ket{\xi_{1,\pi}}\ket{1}\ket{\pi} \ket{\phi_{\sigma,1}^{(\pi)}}\ket{\phi_{\sigma,0}^{(\pi')}} \\
&=& \frac{1}{\sqrt{|S_n|}\sqrt{1-\|\ket{\xi_{1,\pi'}}\|^2}} \sum_{\sigma}\sum_{\pi} \ket{\xi_{1,\pi}}\ket{1}\ket{\pi} \ket{\phi_{\sigma,1}^{(\pi)}}\ket{\phi_{\sigma,0}^{(\pi')}}.
\end{eqnarray*}

The aforementioned claim (*) will be proven as follows.
First, we note that $\| \ket{\eta^{(\pi')}} \|^2$ equals $\omega_n^2 \sum_{\sigma,\tau}\sum_{\pi} \| \ket{\xi_{1,\pi}} \|^2 (\bra{\phi_{\sigma,0}^{(\pi')}}\bra{\phi_{\sigma,1}^{(\pi)}}) (\ket{\phi_{\tau,1}^{(\pi)}}\bra{\phi_{\tau,0}^{(\pi')}})$, which is $\omega_n^2 \sum_{\sigma,\pi} \|\ket{\xi_{1,\pi}}\|^2 [ \sum_{\tau} \measure{\phi_{\sigma,1}^{(\pi)}}{\phi_{\tau,1}^{(\pi)}}  \measure{\phi_{\sigma,0}^{(\pi')}}{\phi_{\tau,0}^{(\pi')}} ]$. By Lemma \ref{base-phi-1}, $\measure{\phi_{\sigma,1}^{(\pi)}}{\phi_{\tau,1}^{(\pi)}} = 1$ if $\tau=\sigma$; $-1$ if $\tau=\sigma\pi$; $0$ otherwise. Moreover, $\measure{\phi_{\sigma,0}^{(\pi')}}{\phi_{\tau,0}^{(\pi')}} = 1$ if $\tau=\sigma$ or $\tau=\sigma\pi'$; $0$ otherwise. Thus, it follows that
$\| \ket{\eta^{(\pi')}} \|^2 = \omega_n^2 \sum_{\pi:\pi\neq \pi'} \| \ket{\xi_{1,\pi}} \|^2 [\sum_{\sigma}  \measure{\phi_{\sigma,1}^{(\pi)}}{\phi_{\sigma,1}^{(\pi)}}  \measure{\phi_{\sigma,0}^{(\pi')}}{\phi_{\sigma,0}^{(\pi')}} ] = \omega_n^2 \sum_{\pi:\pi\neq \pi'} \| \ket{\xi_{1,\pi}} \|^2 \cdot |S_n|$. Since $\sum_{\pi:\pi\neq \pi'} \| \ket{\xi_{1,\pi}} \|^2 = 1 - \|\ket{\xi_{1,\pi'}}\|^2$, we obtain  $\| \ket{\eta^{(\pi')}} \|^2 = \omega_n^2|S_n|(1 - \|\ket{\xi_{1,\pi'}}\|^2)$. By the definition of $\omega_n$, the lemma follows.
\end{proof}

In the subsequent subsection, we will explain how to solve HPSP efficiently on a quantum computer.

\subsection{HPSP Algorithm}\label{sec:HPSP-algorithm}

To solve HPSP, we first generate a quantum state $\ket{\eta^{(\pi')}_{norm}}$ with an appropriate probability, apply $U_2^{(0)}$, and finally measure selected qubits.  The following quantum algorithm $\AAA_{HPSP}$ behaves exactly as described.
In what follows, we tend to drop superscript ``$\UU$'' from $T_s^{(\UU)}$ for brevity.

\s
\n\hrulefill \s\\
\n{\sc HPSP Algorithm $\AAA_{HPSP}$:}
\renewcommand{\labelitemi}{$\circ$}
\begin{itemize}\vs{-1}
  \setlength{\topsep}{-2mm}%
  \setlength{\itemsep}{1mm}%
  \setlength{\parskip}{0cm}%

\item[(M1)] Assume that we are given a quantum state $\rho=\rho_{0}^{(\pi')}$ with an unknown permutation $\pi'\in K_n$. We consider its purification of the form $\ket{\Phi_{0}^{(\pi')}} = \frac{1}{\sqrt{|S_n|}} \sum_{\tau\in S_n} \ket{\tau}\ket{\phi_{\tau,0}^{(\pi')}} = \frac{1}{\sqrt{|S_n|}} \sum_{\tau\in S_n} \ket{\phi_{\tau,0}^{(\pi')}} \ket{\tau}$. Since we are given only a reduced state $\rho$, we assume that we are allowed to manipulate only the first register of $\ket{\Phi_{0}^{(\pi)}}$.
    Starting with $\ket{0}$, we apply $U_1\otimes I$ and then run $\AAA_{dis}$. We then obtain $\sqrt{T_1(n)}\ket{\eta^{(C1)}}$; that is, $\frac{\sqrt{T_1(n)}}{\sqrt{|S_n|}} \sum_{\sigma,\pi}\ket{\xi_{1,\pi}} \ket{1}\ket{\pi} \ket{\sigma} \ket{\phi_{\sigma,1}^{(\pi)}} \otimes \frac{1}{\sqrt{|S_n|}} \sum_{\tau} \ket{\tau}\ket{\phi_{\tau,0}^{(\pi')}}$.

\item[(M2)] Transform $\ket{\pi} \ket{\sigma} \ket{\phi_{\sigma,1}^{(\pi)}}$ into $\ket{\pi} \ket{\sigma} \ket{\mathrm{id}}$ by running $P_2$ and $P_1^{-1}$. Now, we obtain  $\frac{\sqrt{T_1(n)}}{\sqrt{|S_n|}}  \sum_{\sigma,\pi}\ket{\xi_{1,\pi}} \ket{1}\ket{\pi} \ket{\sigma} \ket{\mathrm{id}} \otimes \frac{1}{\sqrt{|S_n|}} \sum_{\tau} \ket{\tau} \ket{\phi_{\tau,0}^{(\pi')}}$.

\item[(M3)] Swap two registers $\ket{\sigma} \ket{\mathrm{id}}$ and $\ket{\tau}\ket{\phi_{\tau,0}^{(\pi')}}$ to obtain $\frac{\sqrt{T_1(n)}}{\sqrt{|S_n|}} \sum_{\tau,\pi}\ket{\xi_{1,\pi}} \ket{1}\ket{\pi} \ket{\tau} \ket{\phi_{\tau,0}^{(\pi')}} \otimes \frac{1}{\sqrt{|S_n|}} \sum_{\sigma} \ket{\sigma} \ket{\mathrm{id}}$.

\item[(M4)] Transform $\ket{\pi}\ket{\tau}$ into $\ket{\pi}\ket{\phi_{\tau,1}^{(\pi)}}$ by applying $P_1$. Moreover, transform $\frac{1}{\sqrt{|S_n|}} \sum_{\sigma} \ket{\sigma} \ket{\mathrm{id}}$ into $\ket{\mathrm{id}}\ket{\mathrm{id}}$.
    The current state is now of the from $\frac{\sqrt{T_1(n)}}{\sqrt{|S_n|}} \sum_{\tau,\pi}\ket{\xi_{1,\pi}} \ket{1}\ket{\pi} \ket{\phi_{\tau,1}^{(\pi)}} \ket{\phi_{\tau,0}^{(\pi')}} \otimes \ket{\mathrm{id}}\ket{\mathrm{id}}$, which equals $\sqrt{T_1(n)(1-\|\ket{\xi_{1,\pi'}}\|^2)} \ket{\eta^{(\pi')}_{norm}}\otimes \ket{\mathrm{id}}\ket{\mathrm{id}}$ by Lemma \ref{norm-of-eta-pi}.

\item[(M5)] Apply $U_2^{(0)}\otimes I$ to the subsystem  $\HH_{A}^{(R1)}\otimes \HH_{B}^{(R1)}$.

\item[(M6)] Measure $\HH_{bit}\otimes \HH_{open1}$. Whenever we observe $(a,\pi)$, if $a\neq 0$, then reject. Otherwise, output $\pi$.
\end{itemize}\vs{-2}
\n\hrulefill\\


To complete the proof of Lemma \ref{HPSP-S0-S1-bound}, it suffices to show that the success probability $p_{\pi'}$ of obtaining $\pi'$ from $\rho_{0}^{(\pi')}$ by running $\AAA_{HPSP}$, over all $\pi\in K_n$ chosen uniformly at random, is at least $\frac{1}{8p(n)^2}$. This statement follows from two separate claims. The first claim below makes a bridge between the probability $p_{\pi'}$ and the state $(I\otimes M_0)(U_2^{(0)}\otimes I)(I\otimes  M_1) \ket{\eta^{(C1)}}$.

\begin{claim}\label{HPSP-success-prob}
For any fixed $\pi'\in K_n$, the success probability $p_{\pi'}$ of obtaining $\pi'$ from $\rho_{0}^{(\pi')}$ by running $\AAA_{HPSP}$ is at least $2(1-\frac{2}{|K_n|+1})^2 \| (I\otimes M_0)(U_2^{(0)}\otimes I)(I\otimes  M_1) \ket{\eta^{(C1)}} \|^2$.
\end{claim}

\begin{proof}
Recall that  $\tilde{M}_{\pi'} \equiv \sum_{\sigma\in S_n^{(\pi')}} M_{commit}^{(0,\sigma,\pi')}$ and $M_a \equiv \sum_{\pi\in K_n} M_{bit}^{(a)}\otimes M_{open1}^{(\pi)}\otimes M_{mix}^{(a,\pi)}$ for each index $a\in\{0,1\}$.
Since two operators $I\otimes \tilde{M}_{\pi'}$ and $U_2^{(0)}\otimes I$ are commutable, it follows that
\[
(I\otimes \tilde{M}_{\pi'})(U_2^{(0)}\otimes I)\ket{\eta^{(C1)}_{perf}} = (U_2^{(0)}\otimes I) (I\otimes \tilde{M}_{\pi'})\ket{\eta^{(C1)}_{perf}} = (U_2^{(0)}\otimes I)\ket{\eta^{(\pi')}}.
\]
Since $I\otimes M_0 = (I\otimes M_0)(I\otimes \tilde{M}_{\pi'})$, we obtain
\[
(I\otimes M_0)(U_2^{(0)}\otimes I)\ket{\eta^{(C1)}_{perf}} = (I\otimes M_0)(I\otimes \tilde{M}_{\pi'})(U_2^{(0)}\otimes I)\ket{\eta^{(C1)}_{perf}}
= (I\otimes M_0)(U_2^{(0)}\otimes I)\ket{\eta^{(\pi')}}.
\]

Through Steps (M1)--(M4), we generates $\sqrt{T_1(n)(1-\|\ket{\xi_{1,\pi'}}\|^2)} \ket{\eta^{(\pi')}_{norm}}$.
From Steps (M5)--(M6), the success probability $p_{\pi'}$ for a fixed $\pi'$ is exactly $\| \sqrt{T_1(n)(1-\|\ket{\xi_{1,\pi'}})\|^2)} (I\otimes M_0) (U_2^{(0)}\otimes I) \ket{\eta^{(\pi')}_{norm}} \|^2$.
By the proof of Lemma \ref{norm-of-eta-pi}, it holds that
$\ket{\eta^{(\pi')}_{norm}} = \frac{\sqrt{2}(|K_n|-1)}{\sqrt{1-\|\ket{\xi_{1,\pi'}}\|^2}(|K_n+1)} \ket{\eta^{(\pi')}}$. Thus, $p_{\pi'}$ equals $\| \frac{\sqrt{2T_1(n)}(|K_n|-1)}{|K_n|+1} (I\otimes M_0) (U_2^{(0)}\otimes I) \ket{\eta^{(C1)}_{perf}} \|^2$, which is $\frac{2(|K_n|-1)^2}{(|K_n|+1)^2} \| (I\otimes M_0) (U_2^{(0)}\otimes I) (I\otimes M_1) \ket{\eta^{(C1)}} \|^2$ since $\sqrt{T_1(n)}\ket{\eta^{(C1)}_{perf}} = (I\otimes M_1)\ket{\eta^{(C1)}}$.

\ignore{
Hence, the average success probability $p = \frac{1}{|K_n|}\sum_{\pi'}p_{\pi'}$ is $\frac{1}{|K_n|} \sum_{\pi'\in K_n} \frac{2(|K_n|-1)^2}{(|K_n|+1)^2} \| (I\otimes M_{0}) (U_2^{(0)}\otimes I)(I\otimes M_{1}) \ket{\eta^{(C1)}} \|^2$. Clearly, this value equals
$\frac{2(|K_n|-1)^2}{(|K_n|+1)^2} \| (I\otimes M_0)(U_2^{(0)}\otimes I)(I\otimes  M_1) \ket{\eta^{(C1)}} \|^2$.
}
\end{proof}

\begin{claim}\label{norm-bound}
$\| (I\otimes M_0)(U_2^{(0)}\otimes I)(I\otimes M_1)\ket{\eta^{(C1)}} \|^2 \geq 1/4p(n)^2$.
\end{claim}

\begin{proof}
For simplicity, we write $\varepsilon$ for $1/p(n)$. It holds that
$\| (I\otimes M_0)(U_2^{(0)}\otimes I)(I\otimes M_1)\ket{\eta^{(C1)}} \|^2 \geq (\sqrt{T_0} - \sqrt{1-T_1})^2$ by an argument similar to \cite[Lemma 4]{DMS00}. Since $T_0 -\varepsilon \geq 1 - T_1$, it follows that $\sqrt{T_0} - \sqrt{1-T_1}\geq \sqrt{T_0} - \sqrt{T_0-\varepsilon} = \sqrt{T_0}(1-\sqrt{1-\frac{\varepsilon}{T_0}})$. Since $\sqrt{1-x}\leq 1-\frac{x}{2}$ for any real number $x\leq 1$, $\sqrt{T_0}(1-\sqrt{1-\frac{\varepsilon}{T_0}}) \geq \sqrt{T_0}\cdot \frac{\varepsilon}{2T_0} = \frac{\varepsilon}{2\sqrt{T_0}} \geq \frac{\varepsilon}{2}$ since $T_0\leq 1$. Hence, we conclude that
$(\sqrt{T_0} - \sqrt{1-T_1} )^2 \geq \frac{\varepsilon^2}{4}$, as requested.
\end{proof}

By combining Claims \ref{HPSP-success-prob} and \ref{norm-bound}, we obtain the desired consequence that the success probability of obtaining $\pi'$ from $\rho_{0}^{(\pi')}$ by running $\AAA_{HPSP}$ is at least $2(1-\frac{2}{|K_n|+1})^2\cdot \frac{1}{4p(n)^2} \geq 2\cdot (\frac{1}{2})^2\cdot  \frac{1}{4p(n)^2} \geq \frac{1}{8p(n)^2}$ for any number  $n \geq 3$.

\section{A Brief Discussion on Our Protocol}

Through Sections \ref{sec:analysis-protocol}--\ref{sec:algorithm-HPSP}, we have shown that our quantum bit commitment scheme achieves computational concealing and statistical binding at communication cost of $O(n\log{n})$, where $n$ is a security parameter, since Alice sends the information on $(a,\sigma,\pi,\ket{\phi_{\sigma,a}^{(\pi)}})$ to Bob during the two phases and the permutations $\sigma$ and $\pi$ are expressed using $O(n\log{n})$ bits. Although our protocol requires a weaker assumption than that of \cite{DMS00}, the communication cost is larger. For a practical application, it is better to reduce the communication cost as in the case of, for example,  \cite{Tan03}.

\ms
\paragraph{Acknowledgements.}
The author is grateful to Paulo Mateus for inviting him to Departamento de Matem{\'a}tica do Instituto Superior T{\'e}cnico da Universidade de Lisboa between March 12--21, 2013, where a discussion with Andr{\'e} Souto, Paulo Mateus, and Pedro Ad{\~a}o inspired this work. He also thanks Marcos Villagra for finding reference \cite{CKR11}.

\let\oldbibliography\thebibliography
\renewcommand{\thebibliography}[1]{%
  \oldbibliography{#1}%
  \setlength{\itemsep}{0pt}%
}
\bibliographystyle{plain}

\begin{thebibliography}{Gur91}
{\small

\bibitem{BB84}
C. H. Bennett and G. Brassard. Quantum cryptography: public key distribution and coin flipping. In {\em Proc. of the IEEE International Conference on Computers, Systems, and Signal Processing}, 1984, pp.175--179.

\bibitem{BV97}
E. Bernstein and U. Vazirani. Quantum complexity theory. {\em SIAM J. Comput.} 26 (1997) 1411--1473.

\bibitem{CK11}
A. Chailloux and I. Kerenidis. Optimal bounds for quantum bit commitment. Available at arXiv:1102.1678v1, February 2011.

\bibitem{CKR11}
A. Chailloux, I. Kerenidis, and B. Rosgen. Quantum commitments from complexity assumptions. Available at arXiv:1010.2793v2, July 2011.

\bibitem{DFR+07}
I. Damg{\aa}rd, S. Fehr, R. Renner, L. Salvail, and C. Schaffner. A tight high-order entropic quantum uncertainty relation with applications. In {\em Advances in Cryptology} - CRYPTO 2007, Lecture Notes in Computer Science, Springer, vol. 4622, pp.360--378, 2007.

\bibitem{DFSS05}
I. Damg{\aa}rd, S. Fehr, L. Salvail, and C. Schaffner. Cryptography in the bounded quantum-storage model. In {\em Proc. of the 46th IEEE Symposium of Foundations of Computer Science} (FOCS 2005), pp.449--458, 2005.

\bibitem{DV12}
A. Danan and L. Vaidman. Practical quantum bit commitment protocol. {\em Quantum Inf. Process.} 11 (2011) 769--775.

\bibitem{DMS00}
P. Dumais, D. Mayers, and L. Salvail. Perfectly concealing quantum bit commitment from any quantum one-way permutation. In {\em Proc.
of the 19th International Conference on Theory and application of cryptographic techniques} (EUROCRYPT'00), Springer, Lecture Notes in Computer Science, vol. 1807, pp.300--315, 2000.

\bibitem{LC97}
H. K. Lo and H. F. Chau. Is quantum bit commitment really possible? {\em Phys. Rev. Lett.} 78 (1997) 3410--3413.

\bibitem{KNY08}
M. Kada, H. Nishimura, and T. Yamakami. The efficiency of quantum identity testing of multiple states. {\em J. Phys. A: Math. Theor.} 41 (2008) article no. 395309.

\bibitem{KKNY12}
A. Kawachi, T. Koshiba, H. Nishimura, and T. Yamakami. Computational indistinguishability between quantum states and its cryptiographic application. {\em Journal of Cryptology} 25 (2012) 528--555. A preliminary version appeared in the {\em Proc. of EUROCRYPT 2005}, Lecture Notes in Computer Science, vol. 3494, pp.268--284, 2005.

\bibitem{KO11}
T. Koshiba and T. Odaira. Non-interactive statistically-hiding quantum bit commitment from any quantum one-way function. Available at arXiv:1102.344v1, 2011.

\bibitem{May97}
D. Mayers. Unconditionally secure quantum bit commitment is impossible. {\em Phys. Rev. Lett.} 78 (1997) 3414--3417.

\bibitem{Tan03}
K. Tanaka. Quantum bit-commitment for small storage based on quantum one-way permutations. {\em New Generation Computing} 21 (2003), 339--345.

\bibitem{Yao93}
A. C. Yao, Quantum circuit complexity. In {\em Proc. of the 34th IEEE Symposium on Foundations of Computer Science}, pp.352--361, 1993.

}
\end{thebibliography}

\end{document}